\documentclass[a4paper,12pt]{article}

\usepackage{titletoc}
\usepackage{fancyhdr}
\usepackage[utf8]{inputenc}
\usepackage{amssymb}
\usepackage{amsmath}
\usepackage{amsthm}
\usepackage{verbatim}
\usepackage{fullpage}
\usepackage[utf8]{inputenc}
\usepackage{setspace}
\usepackage{geometry}
\usepackage{setspace}
 \usepackage{amsmath,amscd}
\usepackage{a4wide}
\usepackage{color}
\usepackage{dsfont}
\numberwithin{equation}{section}

\usepackage{ifpdf}
\ifpdf
    \usepackage[pdftex]{graphicx}   
    \usepackage[pdftex,     
            plainpages=false,   
            breaklinks=true,    
            colorlinks=false,
            pdftitle=My Document
            pdfauthor=My Good Self
           ]{hyperref}
    \usepackage{thumbpdf}
\else
    \usepackage{graphicx}       
    \usepackage{hyperref} {colorlinks=false}      
\fi

\newcommand{\al}{\alpha}
\newcommand{\R}{\mathbb{R}}

\newcommand{\p}{\partial}

\newcommand{\ed}{\color{red}}
\newcommand{\eed}{\color{black}}

\def \qed {\hfill \vrule height6pt width 6pt depth 0pt}

\newtheorem{Remark}{Remark}[section]
\newtheorem{Theorem}{Theorem}[section]
\newtheorem{Lemma}{Lemma}[section]
\newtheorem{Definition}{Definition}[section]
\newtheorem{Assumption}{Assumption}[section]
\newtheorem{Proposition}{Proposition}[section]

\begin{document}

 \title { Combined Mean Field Limit and Non-relativistic Limit of Vlasov-Maxwell Particle System to Vlasov-Poisson System }

\author {Li Chen\footnotemark[1]  \, \quad Xin Li\footnotemark[2] \, \quad Peter Pickl\footnotemark[3] \footnotemark[4] \, \quad Qitao Yin\footnotemark[5]}
\footnotetext[1]{Mathematisches Institut, Universit\"at Mannheim, Germany. E-mail: chen@math.uni-mannheim.de.} \footnotetext[2]{Beijing Information Science and Technology University, China. E-mail: lixin91600@163.com}\footnotetext[3]{Mathematisches Institut, Universit\"at M\"unchen, M\"unchen, Germany. Email: pickl@math.lmu.de}\footnotetext[4]{Duke Kunshan University, Kunshan, Jiangsu, China. Email: peter.pickl@duke.edu}\footnotetext[5]{Tsinghua University High School, Beijing, China. Email: yinqitao@hotmail.com}
\setlength{\baselineskip}{17pt}{\setlength\arraycolsep{2pt}

 \maketitle

\begin{abstract}
In this paper we consider the mean field limit and non-relativistic limit of relativistic Vlasov-Maxwell particle system to Vlasov-Poisson equation. With the relativistic Vlasov-Maxwell particle system being a starting point, we carry out the estimates (with respect to $N$ and $c$) between the characteristic equation of both Vlasov-Maxwell particle model and Vlasov-Poisson equation, where the probabilistic method is exploited. In the last step, we take both large $N$ limit and non-relativistic limit (meaning $c$ tending to infinity) to close the argument.
\end{abstract}
{\bf AMS Classification:} 
\section{Introduction}

The time evolution of plasmas is a very important topic of physics. In many cases, for example when considering the plasma in a nuclear fusion reactor, the temperature of the particles forming the plasma is sufficiently high to neglect quantum effects. Given that the number of particles forming the plasma is very high, also a mean-field approximation for the internal electromagnetic forces of the system can be argued  \cite{Golse2012}, thus the system is in good approximation given by the relativistic Vlasov-Maxwell equations: 
 \begin{equation}
\label{VM}
\begin{cases}
 \displaystyle {\partial _t} f_{m}+\hat{v}\cdot \nabla_x f_{m}+(E_{m}+c^{-1}\hat{v}\times B_{m})\cdot \nabla_vf_{m}=0,\\
 \displaystyle{\partial _t} E_{m}=c\nabla\times B_{m}-j_{m}, \quad \nabla \cdot E_{m}=\rho_{m},\\
  \displaystyle {\partial _t} B_{m}=-c\nabla \times E_{m}, \quad \nabla \cdot B_{m}=0,
\end{cases}
\end{equation}
where  $\displaystyle \hat{v}=\frac{v}{\sqrt{1+c^{-2}v^2}}$, $\displaystyle \rho_{m}(t,x)=
\int_{\mathbb{R}^{ 3}}{f_{m}}(t,x,v)\,dv$ and $\displaystyle j_{m}(t,x)=
\int_{\mathbb{R}^{ 3}}\hat{v}{f_{m}}(t,x,v)\,dv $. The parameter
$c$ is the speed of light, $\displaystyle \left( {{E_{m}},{B_{m}}} \right)$ is the
electro-magnetic field, and the distribution function
$\displaystyle {f}_{m}(t,x,v)\geq 0$ describes the density of particles with
position $\displaystyle x\in{\mathbb{R}^{ 3}}$ and velocity $\displaystyle v\in
{\mathbb{R}^{ 3}}$.  Assuming that there are no external electromagnetic fields, the initial data

\begin{equation}
\label{IVM}
\begin{cases}
  \displaystyle f_{m}(0,x,v) =f_{0}(x,v), \hfill  \\
\displaystyle E_{m}(0,x) =E_{0}(x), \hfill  \\
 \displaystyle B_{m}(0,x) =B_{0}(x),
 \end{cases}
\end{equation}
satisfy the compatibility conditions $\displaystyle \nabla\cdot
E_{0}(x)=\rho_{0}(x)=\int_{\mathbb{R}^{ 3}}{f}_{0}(x,v)\,dv$, $\nabla\cdot B_{0}(x)=0$. \\
 Local existence and uniqueness of classical solutions to this initial value problem for smooth and compactly supported data was established in \cite{GS86}. These solutions can be extended globally in time provided the momentum support can be controlled  assuming certain conditions on the initial data, e.g. smallness \cite{GS87} closeness to neutrality \cite{GS88} or closeness to spherical symmetry \cite{Rein1990}. It is worth to mention that  different approaches to the results in \cite{GS86} were given in \cite{BGP03,KS02}. In order to obtain global existence of solutions, DiPerna and Lions restricted the solution concept to weak solutions. We refer to \cite{DL89}.

  As was shown in \cite{Schaeffer86} using an integral representation for the electric and magnetic field due to Glassey and Strauss \cite{GS86}, the solutions of relativistic Vlasov-Maxwell system converge in  the  point-wise sense to solutions of the non-relativistic Vlasov-Poisson system (below) at the rate of $1/c$ as $c$ tends to infinity. The Vlasov-Poisson system reads
  \begin{equation}
\label{VP}
 \begin{cases}
 \displaystyle {\partial _t} f_{p}+v\cdot \nabla_x f_{p}+E_{p}\cdot \nabla_vf_{p}=0,\\
 \displaystyle E_{p}(t,x)=\frac{1}{4\pi}\int_{\R^3}\rho_{p}(t,y)\frac{x-y}{|x-y|^{3}}\,dy,\\
\displaystyle  \rho_{p}(t,y)=\int_{\R^3}f_{p}(t,y,v)\,dv,
\end{cases}
\end{equation}
with the initial data $ f_{p}(0,x,v) =f_{0}(x,v)$. Here the  indexes $m$ and $p$ in \eqref{VM} and \eqref{VP}  stand for  Maxwell and Poisson respectively. We note that there are global existence results for classical solutions of the Vlasov-Poisson system \cite{LP91,Pfaffelmoser1992,Schaeffer91}.

  However, a more interesting and challenging question  is to consider what  the corresponding particle model of relativistic Vlasov-Maxwell equation might be and whether we can prove  the validity of the  mean field description rigorously in the limit $N\to\infty$. Up to our knowledge and at the time of this writing, taking both the  mean field limit  and the non-relativistic limit (or classical limit) of the Vlasov-Maxwell system into account is rare in the literatures. Concerning the mean field limit, Braun and Hepp \cite{BH77} and Dobrushin \cite{Dobrushin1979} have proposed rigorous derivations of a system analogous to the Vlasov-Poisson system with a twice differentiable mollification of the Coulomb potential. Hauray and Jabin \cite{HJ07} have succeeded in treating the case of singular potentials, but not including the Coulomb singularity yet. Recently, Lazarovici and Pickl \cite{LP17} gave a probabilistic proof of the  validity of the  mean field limit and propagation of chaos for the $N$-particle systems in three dimensions with Coulomb potential  with $N$-dependent cutoff, which provides us with a very constructive idea of method. 
  Lazarovici generalized this result including electromagnetic fields, proving the validity of the relativistic Vlasov-Maxwell equation \cite{Lazarovici2016} considering charges of radius $N^{-\delta}$ with $\delta<1/12$. 

Writing down the corresponding $N$-particle model of the non-relativistic Vlasov-Maxwell system is a perplexing task because one needs to find a suitable description for the electromagnetic self-interaction within the theory of classical electrodynamics \cite{EKR09,Jackson2012,Spohn2004}. The problem of deriving a regularized version of the Vlasov-Maxwell system from a particle model was explicitly mentioned by Kiessling in \cite{Kiessling2008}. Only after several years did Golse \cite{Golse2012} establish the mean field limit of a $N$-particle system towards a regularized version of the relativistic Vlasov-Maxwell system with the help of \cite{EKR09} by Elsken, Kiessling and Ricci. \eed

In the present work, we want to combine the mean field limit and non-relativistic limit of the regularized relativistic Vlasov-Maxwell particle model to Vlasov-Poisson equation. The method we apply here is more or less along the line of \cite{GS86, Golse2012, LP17}  using a mollifier for regularization removes  the difficulties caused by the electromagnetic self-interaction forces. Unlike regularizing the Coulomb potential in the mean field limit established in \cite{BH77, Dobrushin1979}, the regularization of the self-interaction force in the Vlasov-Maxwell system is more difficult since the electromagnetic field involves both a scalar and vector potentials \cite{Golse2012}. The solutions of the relativistic Vlasov-Maxwell system, as was discussed by Glassey and Strauss in \cite{GS86}, are closely related to the wave equation. This connection uses Kirchhoff's formula, which we also used in this paper. We would like to mention that there are other representations of the solutions of the relativistic Vlasov-Maxwell system, for example \cite{BGP03,BGP04}, but they are all in fact equivalent.

This paper is organized as follows: in Section 2, we prepare the regularization-procedure of both, the relativistic Vlasov-Maxwell and the Vlasov-Poisson system, and provide estimates between the solutions of these two systems. In Section 3, we introduce the particle model of the relativistic Vlasov-Maxwell system and apply the probabilisitic method to carry out the estimates between the characteristic equation of the particle model and that of the relativistic Vlasov-Maxwell system. We summarize our results in Section 4.

\section{Regularization of the Vlasov-Maxwell and the Vlasov-Poisson Systems}

Let $\chi \in C^{\infty}_{0}$ satisfy
$$\chi(x)=\chi(-x)\geq 0,~~~supp(\chi)\subset B_1(0),~~ \int_{\mathbb{R}^{3}}\chi(x)\,dx=1,$$
and define the regularizing sequence
\begin{equation}
\label{regularize}{\chi}^{N}(x)={N}^{3\theta}\chi\left(N^{\theta}x\right).
\end{equation}
  The regularized  version  of the Vlasov-Maxwell
 System (VMN)  with unknown $(f^N_{m},B^N_{m},E^N_{m})$  is given by: 
 \begin{equation}
\label{VMN}
\begin{cases}
  \displaystyle   {\partial _t}{f^{N}_{m}} +\hat{v}\cdot\nabla_{x}{f^{N}_{m}}+\left({E^{N}_{m}}+c^{-1}\hat{v}\times{B^{N}_{m}}\right)\cdot\nabla_{v}{f^{N}_{m}}=0, \hfill  \\
  \displaystyle  {\partial _t}{E^{N}_{m}}=c\nabla\times {B^{N}_{m}}- {\chi}^{N}\ast_{x}{\chi}^{N}\ast_{x}j^{N}_{m},~~\quad
    \nabla\cdot E^{N}_{m}={\chi}^{N}\ast_{x}{\chi}^{N}\ast_{x}\rho^{N}_{m},\hfill  \\
   \displaystyle  {\partial _t} {B^{N}_{m}}=-c\nabla\times {E^{N}_{m}},~~\quad \quad \quad \nabla\cdot B^{N}_{m}=0,
\end{cases}
\end{equation}
and initial data (IVMN)
 \begin{equation}
\label{IVMN}
 \begin{cases}
  \displaystyle f^{N}_{m}(0,x,v) =f_{0}(x,v), \hfill  \\
   \displaystyle E^{N}_{m}(0,x) ={\chi}^{N}\ast_{x}{\chi}^{N}\ast_{x}E_{0}(x), \hfill  \\
  \displaystyle B^{N}_{m}(0,x) ={\chi}^{N}\ast_{x}{\chi}^{N}\ast_{x}B_{0}(x),
 \end{cases}
\end{equation}
The regularized version  of the Vlasov-Poisson
 System (VPN) with unknown $(f^N_{p},E^N_{p})$ is given by:

\begin{equation}
\label{VPN}
\begin{cases}
 \displaystyle {\partial _t} f^{N}_{p}+v\cdot \nabla_x f^{N}_{p}+E^{N}_{p}\cdot \nabla_vf^{N}_{p}=0,\\
 \displaystyle E^{N}_{p}(t,x)=\frac{1}{4\pi}\iiint_{\R^9}{\frac{x-y}{|x-y|^{3}}{\chi}^{N}(p){\chi}^{N}(z)\rho^{N}_{p}(t,y-p-z)\,dydpdz},\\
\displaystyle  \rho^{N}_{p}(t,y)=\int_{\R^3}{f^{N}_{p}(t,y,v)\,dv},
\end{cases}
\end{equation}
with the initial data $ f^{N}_{p}(0,x,v) =f_{0}(x,v)$.
\begin{Theorem}\label{T1}
Let $f_{0}$ be a nonnegative $C^1$-function with compact support in
$\mathbb{R}^6$ and  $ B_{0}$ be in $C^{2}_{0}(\mathbb{R}^3)\cap
W^{1,\infty}(\mathbb{R}^{3})\cap W^{2,1}(\mathbb{R}^3)$.  Assume further that
$$\|\nabla_{x}B_{0}\|_{L^{\infty}(\R^3)}+\|B_{0}\|_{L^{\infty}(\R^3)}\leq \frac{1}{c^2},$$ and
$$E_{0}(x)=\frac{1}{4\pi}\iint_{\R^6}{\frac{x-y}{|x-y|^{3}}f_{0}(y,v)\,dvdy}.$$
Then 
\begin{enumerate}
	\item There exists a $\overline{T}>0$ such that (VPN) with initial data $f^{N}_p(0,x,v)=f_{0}(x,v)$ admits a
unique $C^1$-solution $(f^{N}_p, E^{N}_p)$ on the time interval $[0,\overline{T})$.
 
 \item  There exists  a $T^{\ast}>0$ (independent of c) such that for
$c\geq1,$ (VMN) with the initial condition (IVMN) has a unique $C^1$-
solution $(f^N_{m},B^N_{m},E^N_{m})$ on the time interval
$[0,T^\ast)$. Furthermore there exist nondecreasing
functions (independent of $c$ and $N$) $q(t):[0,T^\ast)\rightarrow{\mathbb{R}}$
and $H(t):[0,T^\ast)\rightarrow{\mathbb{R}}$ such that
$$ f^{N}_{m}=0, \quad \hbox{if} \quad |v|\geq{q(t)},$$
$$\| E^{N}_{m}(t,x)\|_{L^{\infty}([0,T^\ast)\times \R^3)}+\|B^{N}_{m}(t,x)\|_{L^{\infty}([0,T^\ast)\times \R^3)} \leq{H(t)}.$$

\item Let $\widetilde{T}=\min(\overline{T},T^\ast),$ then for every $T
\in[0,\widetilde{T})$ there exists a constant M (depending on $T$ and the
initial data, but not on $c$) such that for $c\geq 1$
$$\|f^N_{m}-f^N_{p}\|_{L^{\infty}([0,T)\times \R^3 \times \R^3)}+\|E^N_{m}-E^N_{p}\|_{L^{\infty}([0,T)\times \R^3 )}+\|B^N_{m}\|_{L^{\infty}([0,T)\times \R^3 )} \leq \frac{M}{c}.$$
\end{enumerate}
\end{Theorem}
The proof of the Theorem involves  many  long and tedious calculations since it includes many  cut-offs  and mollifications, however, no technical difficulties other than presented in the paper \cite{Schaeffer86} appear.  Therefore we deliver the proof in the  appendix at the end of the paper for those readers who  want to take  a closer look at the details.
 
\begin{Remark}
	In the current setting, i.e., repulsive particle interactions, both $\overline{T}$ and $T^\ast$ can be global. But in the attractive case, there might be lack of global existence of solutions. Therefore both existence results we give are  locally in time. The limits $N \to \infty$ and $c \to \infty$ in our paper are taken in the time interval where both solutions exist.
\end{Remark}
We assume in Theorem 2.1 that $f_0$ has compact support, so let
$$q_{0}=\sup\{|v|: \hbox{there exists} \,x \in {\mathbb{R}^3}\, \hbox{such that}\, f_0(x,v)\neq0\}.$$
 Further, we define the characteristic curves
$(x(t, x_0,v_0,t_0),v(t, x_0,v_0,t_0))$ (or in short $(x(t),v(t))$) by
\begin{eqnarray}
\label{CHC}
\begin{cases}
	\displaystyle
 \frac{dx}{dt} =\hat{v}, \vspace{1mm}\\
\displaystyle \frac{dv}{dt}
=E^{N}_{m}+c^{-1}\hat{v}\times{B^{N}_{m}}.
\end{cases}
\end{eqnarray}
Therefore $f^{N}_{m}$
remains non-negative if $f_0$ is non-negative and
$$\sup\{f^{N}_{m}(t,x,v): x \in{\mathbb{R}^3}, v \in {\mathbb{R}^3}, t\in [0,T^\ast)\}=\|f_0\|_{L^{\infty}( \R^3 \times \R^3)} .$$
We also define
$$p_0=\sup\{|x|: \hbox{there exists} \,v \in {\mathbb{R}^3}\, \hbox{such that}\, f_0(x,v)\neq0\}.$$
Hence $f^{N}_{m}(t,x,v)=0$, if $|x|\geq {p_0+tq(t)}$.\\
 Before we prove the Theorem, we write the second order form of
Maxwell's equation:
\begin{eqnarray}
\label{MWN}
 \begin{cases}
 \displaystyle {\partial _{tt}} E^{N}_{m}-c^2\Delta E^{N}_{m}=-{\chi}^{N}\ast_{x}{\chi}^{N}\ast_{x}(c^2\nabla_{x}\rho^{N}_{m}+ {\partial _{t}}j^{N}_{m}),\\
\displaystyle {\partial _{tt}} B^{N}_{m}-c^2\Delta B^{N}_{m}=c{\chi}^{N}\ast_{x}{\chi}^{N}\ast_{x}\nabla \times{j^{N}_{m}},\\
 \displaystyle E^{N}_{m}(0,x)={\chi}^{N}\ast_{x}{\chi}^{N}\ast_{x}E_0,\\
\displaystyle
B^{N}_{m}(0,x))={\chi}^{N}\ast_{x}{\chi}^{N}\ast_{x}B_0,\\
\p_t E_m^N(0,x)=c\nabla\times {B^{N}_{m}}(0,x)- {\chi}^{N}\ast_{x}{\chi}^{N}\ast_{x}j^{N}_{m}(0,x),\\
\p_t B_m^N(0,x)=-c\nabla\times {E^{N}_{m}}(0,x).
\end{cases}
\end{eqnarray}
\begin{Proposition}
Let $Y(t,x)\in {\mathcal{D'}(\mathbb{R}\times{\mathbb{R}^{3}})}$
satisfy
\begin{equation}
\label{PO}
\begin{cases}
\displaystyle {\partial _{tt}}Y-c^2\Delta{Y}=\delta_{(t,x)=(0,0)}, \hfill \\
\displaystyle
\hbox{supp}\, {Y}\subset{\{(t,x)\in{\mathbb{R}_+\times{\mathbb{R}^{3}}},~~|x|\leq{ct}\}},\hfill
\end{cases}
\end{equation}
 then $\displaystyle Y(t,x)=\frac{\mathds{1}_{t>0}}{4\pi c|x|}\delta(|x|-ct)$.
 \end{Proposition}
  The proof of this Proposition is standard. $Y(t,x)$ is called the fundamental solution of  the  wave equation.
Set $Y^N={\chi}^{N}\ast_{x}{\chi}^{N}\ast_{x}Y,$ then the solutions
of \eqref{MWN}  are  given in terms of
\begin{eqnarray}
\label{SMWN1}
\begin{cases}
 \displaystyle  E^{N}_{m}={\partial _t} Y^N\ast_{x}E_0+Y^N\ast_{x}{(c\nabla\times B_0-j^{N}_{m}(0,\cdot))}
 -Y^N\ast_{t,x}(c^2\nabla{\rho}^{N}_{m}+{\partial _t} j^{N}_{m}),\\
\displaystyle  B^{N}_{m} ={\partial _t}
Y^N\ast_{x}B_0-cY^N\ast_{x}{\nabla\times E_0}
 +cY^N\ast_{t,x}\nabla \times j^{N}_{m}.
\end{cases}
\end{eqnarray}

Using that
$$
\int_{|y-x|\leq{ct}}{h(ct-|y-x|,y)}\,dy=c^2\int^{t}_{0}\int_{|\omega|=1}{(t-\tau)^2
h(c\tau,x+c(t-\tau)\omega)\,d\omega d(c\tau)},
$$
 we can also write the solutions of \eqref{MWN}  in the form 
\begin{eqnarray}
\label{SMWN2}
\begin{cases}
 \displaystyle  E^{N}_{m}=\mathbb{E}_{0}-\frac{1}{4\pi
 c^2}\iiint_{\R^9} \,dpdzdy {\chi}^{N}(p){\chi}^{N}(z)\frac{(c^2\nabla_{y}{\rho}^{N}_{m}+{\partial _t} j^{N}_{m})( t-c^{-1}|x-y|,y-p-z)}{|x-y|},\\
\displaystyle  B^{N}_{m} =\mathbb{B}_{0}
 +\frac{1}{4\pi
 c}\iiint_{\R^9} \,dpdzdy {\chi}^{N}(p){\chi}^{N}(z) \frac{\nabla_{y} \times j^{N}_{m}(y-p-z, t-c^{-1}|x-y|)}{|x-y|},
\end{cases}
\end{eqnarray}
where
\begin{eqnarray}
\begin{cases}
 \displaystyle \mathbb{E}_{0}=\partial_{t}\int_{|\omega|=1}{\frac{t}{4\pi} E^{N}_{m}(0,x+ct\omega)\,d\omega}+\frac{t}{4\pi}\int_{|\omega|=1}
 { \partial_{t}E^{N}_{m}(0,x+ct\omega)\,d\omega},\\
\displaystyle
\mathbb{B}_{0}=\partial_{t}\int_{|\omega|=1}{\frac{t}{4\pi}B^{N}_{m}(0,x+ct\omega)\,d\nu}+\frac{t}{4\pi}\int_{|\omega|=1}{
\partial_{t}B^{N}_{m}(0,x+ct\omega)\,d\omega}.
\end{cases}
\end{eqnarray}

\section{Combined Mean Field Limit and Non-relativistic Limit} \label{sec:model}
\subsection{Regularized Vlasov-Maxwell Particle System}

The regularized Vlasov-Maxwell system  is given by 
 \begin{equation}
\label{VMN}
 \left\{\begin{split}
  & \displaystyle   {\partial _t}{f^{N}_{m}} +\hat{v}\cdot\nabla_{x}{f^{N}_{m}}+\left({E^{N}_{m}}+c^{-1}\hat{v}\times{B^{N}_{m}}\right)\cdot\nabla_{v}{f^{N}_{m}}=0, \hfill  \\
  & \displaystyle  {\partial _t}{E^{N}_{m}}=c\nabla\times {B^{N}_{m}}- {\chi}^{N}\ast_{x}{\chi}^{N}\ast_{x}j^{N}_{m},~~\quad
    \nabla\cdot E^{N}_{m}={\chi}^{N}\ast_{x}{\chi}^{N}\ast_{x}\rho^{N}_{m},\hfill  \\
   &\displaystyle  {\partial _t} {B^{N}_{m}}=-c\nabla\times {E^{N}_{m}},~~\quad \quad \quad \nabla\cdot B^{N}_{m}=0,
\end{split}\right.
\end{equation}
where
$\displaystyle \hat{v}=\frac{v}{\sqrt{1+c^{-2}v^2}}$, $\displaystyle \rho^{N}_{m}(t,x)=\int_{\mathbb{R}^{
3}}{f^{N}_{m}}(t,x,v)\,dv$, $\displaystyle  j^{N}_{m}(t,x)=\int_{\mathbb{R}^{
3}}\hat{v}{f^{N}_{m}}(t,x,v)\,dv$
and the initial data
\begin{equation}
\label{IVMN}
  \left\{\begin{split}
  &\displaystyle f^{N}_{m}(0,x,v) =f_{0}(x,v), \hfill  \\
   &\displaystyle E^{N}_{m}(0,x) ={\chi}^{N}\ast_{x}{\chi}^{N}\ast_{x}E_{0}(x), \hfill  \\
  &\displaystyle B^{N}_{m}(0,x) ={\chi}^{N}\ast_{x}{\chi}^{N}\ast_{x}B_{0}(x),
 \end{split}\right.
\end{equation}
 satisfy the compatibility conditions  $\nabla\cdot
E_{0}(x)=\rho^{N}_{m}(0,x)=\rho_{0}(x),~~\nabla\cdot B_{0}(x)=0.$

 We consider the corresponding interacting particle system with position $x_i \in \R^3$
and velocity $v_i \in \R^3, i=1, \ldots, N$.
The equations of  the characteristics read
 \begin{align}
\label{NS} \begin{cases}
 \displaystyle
 \frac{d}{dt} x_i=\hat{v}(v_i)=\frac{v_i}{\sqrt{1+c^{-2}v_{i}^2}},\\
\displaystyle \frac{d}{dt}v_i=E^{N}_{m}(t,x_i)+c^{-1}\hat{v}(v_i)\times{B^{N}_{m}(t,x_i)},
 \end{cases}
 \end{align}
where 
\begin{equation}
\begin{cases}
 \displaystyle  E^{N}_{m}={\partial _t} Y^N\ast_{x}E_0+Y^N\ast_{x}{(c\nabla\times B_0-j^{N}_{m}(0,.))}
 -Y^N\ast_{t,x}(c^2\nabla{\rho}^{N}_{m}+{\partial _t} j^{N}_{m}),\\
\displaystyle  B^{N}_{m} ={\partial _t}
Y^N\ast_{x}B_0-cY^N\ast_{x}{\nabla\times E_0}
 +cY^N\ast_{t,x}\nabla \times j^{N}_{m}.
\end{cases}
\end{equation}


 Before we 
 present the analytical results in this section, we introduce the following notations:
\begin{Definition}
  \begin{enumerate}
    \item  For any $1\leq i \leq N$ (labeling the particle with position $x_m^{i,N} \in \R^3$ and velocity $v_m^{i,N} \in \R^3$) we denote the pair-interaction force by 
    \begin{eqnarray*}
    &&F_m^{1,N}(t, x_m^{i,N})=-\frac{1}{N-1}\sum_{j=1,j\neq i}^N\int_0^t (\hat{v}(v_m^{j,N}(s))\p_t+c^2\nabla_x )Y^N(t-s,x_m^{i,N}(t)-x_m^{j,N}(s))\,ds,\\     
    && F_m^{2,N}(t, x_m^{i,N},v_m^{i,N})=-\frac{1}{N-1}\sum_{j=1,j\neq i}^N\int_0^t \hat{v}(v_m^{i,N}(t))\times \Big(\hat{v}(v_m^{j,N}(s))\times \\
    && \hspace{8cm}\nabla_x Y^N(t-s,x_m^{i,N}(t)-x_m^{j,N}(s)) \Big)\,ds,
    \end{eqnarray*}
and the mean-field force of the Vlasov system by
 \begin{eqnarray*}
    F_m^{3,N}(t, x_m^{i,N},v_m^{i,N}) &=& E_0\ast_{x}{\partial _t} Y^N(t, x_m^{i,N})+{(c\nabla\times B_0-j^{N}_{m}(0,\cdot))}\ast_{x}Y^N(t, x_m^{i,N})\\
    	&&+c^{-1}\hat{v}(v_m^{i,N})\times B_0\ast_{x}{\partial _t}Y^N(t,x_m^{i,N})\\
    	&&-\hat{v}(v_m^{i,N})\times (\nabla \times E_0)\ast_{x}Y^N(t,x_m^{i,N}).
    \end{eqnarray*}

    \item Let $(X_m^N(t), V_m^N(t))$ be the trajectory on $\R^{6N}$ which evolves according to the Newtonian equation of motion for  the  regularized Vlasov-Maxwell system, i.e.,
    \begin{align}
\label{VMNPF} \begin{cases}
 \displaystyle \frac{d}{dt}X_m^N(t)=\hat{V}(V_m^N(t)), \\ \vspace{0.005cm}\\
 \displaystyle{\frac{d}{dt}}V_m^N(t)=\Psi_m^{1,N}(t,X_m^N(t),V_m^N(t))+\Psi_m^{2,N}(t,X_m^N(t), V_m^N(t))+\Gamma_m^{N}(t, X_m^N(t), V_m^N(t)),
 \end{cases}
 \end{align}
    where $\Psi_m^{1,N}(t,X_m^N(t),V_m^N(t))$ and $\Psi_m^{2,N}(t,X_m^N(t), V_m^N(t))$ denote the total interaction force with
    \begin{eqnarray*}
    	\displaystyle &&\big(\Psi_m^{1,N}(t,X_m^N(t),V_m^N(t))\big)_i =F_m^{1,N}(t,x_m^{i,N})\\
    	&=&-\frac{1}{N-1}\sum_{j=1,j\neq i}^N\int_0^t (\hat{v}(v_m^{j,N}(s))\p_t+c^2\nabla_x )Y^N(t-s,x_m^{i,N}(t)-x_m^{j,N}(s))\,ds, \\
    \displaystyle && \big(\Psi_m^{2,N}(t,X_m^N(t), V_m^N(t))\big)_i =F_m^{2,N}(t,x_m^{i,N},v_m^{i,N})\\
    &=&-\frac{1}{N-1}\sum_{j=1,j\neq i}^N\int_0^t \hat{v}(v_m^{i,N}(t))\times \Big(\hat{v}(v_m^{j,N}(s))\times  \nabla_x Y^N(t-s,x_m^{i,N}(t)-x_m^{j,N}(s)) \Big)\,ds
    \end{eqnarray*}
    while $\Gamma_m^{N}(t, X_m^N(t), V_m^N(t))$ stands for the self-driven force with
    \begin{eqnarray*}
    && \big(\Gamma_m^{N}(t, X_m^N(t), V_m^N(t))\big)_i \\
    &=& F_m^{3,N}(t,x_m^{i,N},v_m^{i,N}) \\
    &=& E_0^N\ast_{x}{\partial _t} Y^N(t,x_m^{i,N})+{(c\nabla\times B^N_0-j^{N}_{m}(0,\cdot))}\ast_{x}Y^N(t,x_m^{i,N})\\
    	&&+c^{-1}\hat{v}(v_m^{i,N})\times B^N_0\ast_{x}{\partial _t}Y^N(t,x_m^{i,N})-\hat{v}(v_m^{i,N})\times (\nabla\times E_0^N)\ast_{x}Y^N(t,x_m^{i,N}).
    \end{eqnarray*}

    \item Let $(\overline{X}_m^N(t), \overline{V}_m^N(t))$ be the trajectory on $\R^{6N}$ which evolves according to the regularized Vlasov-Maxwell equation
    \begin{eqnarray} \label{VMN}
     {\partial _t}{f^{N}_{m}} +\hat{v}\cdot\nabla_{x}{f^{N}_{m}}+\left({E^{N}_{m}}+c^{-1}\hat{v}\times{B^{N}_{m}}\right)\cdot\nabla_{v}{f^{N}_{m}}=0,
    \end{eqnarray}

    i.e.,
    \begin{align}
\label{VMNF} \begin{cases}
 \displaystyle \frac{d}{dt}\overline{X}_m^N(t)=\hat{V}(\overline{V}_m^N(t)),  \\ \vspace{0.005cm}\\
 \displaystyle{\frac{d}{dt}}\overline{V}_m^N(t)=\overline{\Psi}_m^{1,N}(t,\overline{X}_m^N(t))+\overline{\Psi}_m^{2,N}(t,\overline{X}_m^N(t), \overline{V}_m^N(t))+\Gamma_m^{N}(t,\overline{X}_m^N(t), \overline{V}_m^N(t)),
 \end{cases}
 \end{align}
    where
    \begin{eqnarray}
    \label{F1mbar} \displaystyle \big(\overline{\Psi}_m^{1,N}(t,\overline{X}_m^N(t))\big)_i&=&\overline{F}_m^{1,N}(t,\overline{x}_m^{i,N})=-(c^2\nabla{\rho}^{N}_{m}+{\partial _t} j^{N}_{m})\ast_{t,x}Y^N(t,\overline{x}_m^{i,N})\\
    &=&-\iint_{\R^6}\int_0^t\,dsdydv \nonumber \\
    &&\hspace{1.5cm}(c^2\nabla+\hat{v}(v)\p_s )f^N_m(s,\overline{x}_m^{i,N}-y,v)Y^N(t-s,y),\nonumber\\
    \displaystyle   \big(\overline{\Psi}_m^{2,N}(t,\overline{X}_m^N(t), \overline{V}_m^N(t))\big)_i &=& \overline{F}_m^{2,N}(t,\overline{x}_m^{i,N},\overline{v}_m^{i,N})\nonumber\\
    &=&-\iint_{\R^6}\int_0^t \,dsdydv\nonumber\\
    && \label{F2mbar} \hspace{0.5cm}\hat{v}(\overline{v}_m^{i,N})\times \hat{v}(v)\times \nabla_x f^N_m(s,\overline{x}_m^{i,N}-y,v)Y^N(t-s,y),\\
     \big(\Gamma_m^{N}(t, \overline{X}_m^N(t), \overline{V}_m^N(t))\big)_i &=& \overline{F}_m^{3,N}(t,\overline{x}_m^{i,N},\overline{v}_m^{i,N}) \nonumber\\
    &=& E_0^N\ast_{x}{\partial _t} Y^N(t,\overline{x}_m^{i,N})+{(c\nabla\times B^N_0-j^{N}_{m}(0,\cdot))}\ast_{x}Y^N(t,\overline{x}_m^{i,N})\nonumber\\
    	&&+c^{-1}\hat{v}(\overline{v}_m^{i,N})\times B_0^N\ast_{x}{\partial _t}Y^N(t,\overline{x}_m^{i,N})\nonumber\\
    	&&-\hat{v}(\overline{v}_m^{i,N})\times (\nabla \times E_0^N)\ast_{x}Y^N(t,\overline{x}_m^{i,N}).\nonumber
    \end{eqnarray}
 represent the total interaction forces and the self-driven force, respectively.
    \end{enumerate}
\end{Definition}

  $(X(t), V(t))$ and $(\overline{X}(t), \overline{V}(t))$ without superscript $N$ 
  denote the particle configurations driven by the force without cut-off.
   $(X, V)$ and $(\overline{X}, \overline{V})$,  without the argument $t$, stand for  the stochastic initial data, which are independent and identically distributed.
 Note that we always consider the same initial data for both systems, that means $(X, V) = (\overline{X},\overline{V})$. The following lemma gives us  and  estimates on the interaction forces, which will be used in the limiting procedure.

\begin{Lemma}
Let  $\overline{F}_m^{1,N}(t,x)$ and 
$\overline{F}_m^{2,N}(t,x,v)$  be defined as in \eqref{F1mbar} and \eqref{F2mbar}. Then there exists a constant $M$ such that
$$ \|\overline{F}_m^{1,N}(t,x)\|_{L^{\infty}([0,T]\times \R^3)}+ \|\overline{F}_m^{2,N}(t,x,v)\|_{L^{\infty}([0,T]\times \R^3 \times \R^3 )} \leq cM.$$
\end{Lemma}
\begin{proof} By definition, we know that
	\begin{eqnarray}
    &&\displaystyle \|\overline{F}_m^{1,N}(t,x)\|_{L^{\infty}([0,T]\times \R^3)}\nonumber \\
    &=&\|-(c^2\nabla{\rho}^{N}_{m}+{\partial _t} j^{N}_{m})\ast_{t,x}Y^N(t,x)\|_{L^{\infty}([0,T]\times \R^3)} \nonumber \\
    &=&\Big\|\iint_{\R^6}\int_0^t(c^2\nabla_x+\hat{v}(v)\p_s )f^N_m(s,x-y,v)Y^N(t-s,y)\,dsdydv\Big\|_{L^{\infty}([0,T]\times \R^3)},\label{F1Nest}
     \end{eqnarray} 
   where for $s<t$
    \begin{eqnarray*}
    	Y^N(t-s,y)&=&\int_{\R^3}\int_{|z|\leq c(t-s)}\frac{1}{4\pi c |z|}\delta(|z|-c(t-s))\chi^N(y-p-z)\chi^N(p)\,dzdp\\
    	&=& \int_{\R^3}\int_{|\omega|=1}\int_{0}^{c(t-s)}\frac{\tau}{4\pi c}\delta(\tau-c(t-s))\chi^N(y-p-\tau \omega)\chi^N(p)\,d\tau d\omega dp\\
    	&=&  \int_{\R^3}\int_{|\omega|=1}\frac{t-s}{4\pi }\chi^N(y-p-c(t-s) \omega)\chi^N(p)\, d\omega dp,
    \end{eqnarray*}
   and
    \begin{eqnarray*}
    \|	Y^N(t-s,y) \|_{L^{\infty}([0,T]\times \R^3)} \leq \frac{M}{4\pi c} \int_{\R^3}\chi^N(p)\, dp =\frac{M}{4\pi c}.
    \end{eqnarray*}

   So 
     \begin{eqnarray*}
   && \eqref{F1Nest} \\
    &=&\Big\|\frac{1}{4\pi}\iiint_{\R^9}dpdydv \int_0^t ds\int_{|\omega|=1} d\omega  \\
    && \hspace{1.5cm} (t-s)(c^2\nabla_x+\hat{v}(v)\p_s )f^N_m(s,x-y,v)\chi^N(y-p-c(t-s) \omega)\chi^N(p)\Big\|_{L^{\infty}([0,T]\times \R^3)}\\
    &\leq &cM\Big(\sup_{0\leq t\leq T}\|\p_t f^N_m(t,\cdot,\cdot)\|_{L^{\infty}(\R^3\times \R^3)}+\sup_{0\leq t\leq T}\|\nabla_xf^N_m(t,\cdot,\cdot)\|_{L^{\infty}(\R^3\times \R^3)}\Big).
    \end{eqnarray*}
      And similarly we have
    \begin{eqnarray*}
    && \displaystyle  \|\overline{F}_m^{2,N}(t,x,v)\|_{L^{\infty}([0,T]\times \R^3 \times \R^3)} \\
    &=&\Big\|-\iint_{\R^6}\int_0^t \hat{v}(v)\times \hat{v}(z)\times  \nabla_x f^N_m(s,x-y,z)Y^N(t-s,y)\,dsdydz\Big \|_{L^{\infty}([0,T]\times \R^3 \times \R^3)}\\
    &\leq & cM\Big(\sup_{0\leq t\leq T}\|\nabla_xf^N_m(t,\cdot,\cdot)\|_{L^{\infty}(\R^3\times \R^3 )}\Big).
    \end{eqnarray*}
\end{proof}

%

\subsection{Regularized Vlasov-Poisson Particle Model}\label{sec:vppm}
In this section, we consider  the  Vlasov-Poisson particle model and deduce estimates  of the distance between the solutions of Vlasov-Maxwell and Vlasov-Poisson

The regularization of the Vlasov-Poisson
 System (VPN) with unknown $(f^N_{p},E^N_{p})$  is given by

\begin{equation}
\label{VPN}
\begin{cases}
 \displaystyle {\partial _t} f^{N}_{p}+v\cdot \nabla_x f^{N}_{p}+E^{N}_{p}\cdot \nabla_vf^{N}_{p}=0,\\
 \displaystyle E^{N}_{p}(t,x)=\frac{1}{4\pi}\iiint_{\R^9}\,dydpdz \frac{x-y}{|x-y|^{3}}{\chi}^{N}(p){\chi}^{N}(z)\rho^{N}_{p}(t,y-p-z),\\
\displaystyle  \rho^{N}_{p}(t,y)=\int_{\R^3}f^{N}_{p}(t,y,v)\,dv,
\end{cases}
\end{equation}
with the initial data $ f^{N}_{p}(x,v,0) =f_{0}(x,v)$.
Thus, the corresponding Vlasov-Poisson equations of characteristics read
 \begin{equation}
\label{VPCHA} 
\begin{cases}
 \displaystyle \frac{d}{dt} \overline{x}_p^{N}=\overline{v}_p^{N},\hfill \\
\displaystyle \frac{d}{dt}\overline{v}_p^{N}=E^{N}_{p}(t,\overline{x}_p^{N})=\frac{1}{4\pi}\iiint_{\R^9} \,dydpdz \frac{\overline{x}_p^{N}-y}{|\overline{x}_p^{N}-y|^{3}}{\chi}^{N}(p){\chi}^{N}(z)\rho^{N}_{p}(t,y-p-z).
\end{cases}
\end{equation}

Now we can compare the  solutions of the Vlasov-Maxwell and Vlasov-Poisson equations. This requires a more detailed estimate on the respective solutions, which we denote by $f_m^N$ and $f_p^N$.   Using the results in section 3, we know  that 
\begin{eqnarray} \label{MPestimate}
	\|E_m^N+c^{-1}\hat{v}\times B_m^N-E^{N}_{p}\|_{L^{\infty}(\R^3)} &\leq& \|E_m^N-E^{N}_{p}\|_{L^{\infty}(\R^3)}+\|c^{-1}\hat{v}\times B_m^N\|_{L^{\infty}(\R^3)} \nonumber \\
	&\leq &c^{-1}M.
\end{eqnarray}

  We will next  compare the  $N$-particle Vlasov-Maxwell equation with  the Vlasov-Poisson equation. Since the $N$-body system is subject to a regularized force it is most natural to introduce that regularization also for the Vlasov-Poisson system. The translation of the one-body Vlasov-Poisson system to an $N$-body dynamics is straight forward: each particle moves with the same flow given by the Vlasov-Poisson equation. This allows now comparison with the $N$-body characteristics coming from the  $N$-particle Vlasov-Maxwell equation. 
\begin{Definition}
	Let $(\overline{X}_p^N(t), \overline{V}_p^N(t))$ be the trajectory on $\R^{6N}$ which evolves according to the regularized Vlasov-Poisson equation
    \begin{eqnarray} \label{VPN}
    \displaystyle {\partial _t} f^{N}_{p}+v\cdot \nabla_x f^{N}_{p}+E^{N}_{p}\cdot \nabla_vf^{N}_{p}=0,
    \end{eqnarray}

    i.e.,
    \begin{align}
\label{VPNF} \begin{cases}
 \displaystyle \frac{d}{dt}\overline{X}_p^N(t)=\overline{V}_p^N(t),  \\ \vspace{0.005cm}\\
 \displaystyle{\frac{d}{dt}}\overline{V}_p^N(t)=\overline{\Psi}_p^N(t,\overline{X}_p^N(t)),
 \end{cases}
 \end{align}
    where
    \begin{eqnarray*}
    \displaystyle \big(\overline{\Psi}_p^N(t,\overline{X}_p^N(t))\big)_i&=&\overline{F}_p^{N}(t,\overline{x}_p^{i,N})=E^{N}_{p}(t,\overline{x}_p^{i,N})\\
    &=&\frac{1}{4\pi}\iiint_{\R^9}\,dydpdz  \frac{\overline{x}_p^{i,N}-y}{|\overline{x}_p^{i,N}-y|^{3}}{\chi}^{N}(p){\chi}^{N}(z)\rho^{N}_{p}(t,y-p-z).
    \end{eqnarray*}
\end{Definition}


\subsection{Estimates for the Mean Field Limit} \label{sec:meanest}

In this section, we present our key results in full detail. To show the desired convergence, our method can be summarized as follows.
First, we start from the Newtonian system with carefully chosen cut-off and meanwhile introduce an intermediate system which
involves \ed a \eed convolution-type interaction with cut-off,  respectively mollifier given in \eqref{regularize}. Then, we show  convergence of the intermediate system to
the final mean field limit, where the law of large number comes into play. The
crucial point of this method is that we apply stochastic initial data or in other
words we consider a stochastic process.
 This enables us to use tools  from
probability theory, which helps to better understand the mean field process.
The overall procedure can be summarized as follows:


\noindent The following assumptions are used throughout this section.

\begin{Assumption}
  We assume that
  \begin{enumerate}

    \item[(a)]$E_0$ and $B_0$ are all Lipschitz continuous functions.

    \item[(b)] $\displaystyle \al \in \left(0, \frac{1}{8} \right)$, $\displaystyle \beta \in \left(\al, \frac{1-\al}{4}\right)$ and $\displaystyle \theta \in \left(0,\frac{1-\al-4\beta}{16}\right)$.
\end{enumerate}
\end{Assumption}

%

\begin{Definition}
  Let $S_t : \R^{6N} \times \R \to \R$ be the stochastic process given by
  $$S_t=\min\Big\{1, N^{\al}\sup_{0\le s \le t}\Big|(X_m^N(s),V_m^N(s))-(\overline{X}_m^N(s), \overline{V}_m^N(s))\Big|_{\infty}\Big\}.$$
  The set, where $|S_t|=1$, is defined as $\mathcal{N}_{\al}$, i.e.,
 \begin{eqnarray} \label{nalpha}
\mathcal{N}_{\al}:=\left\{(X,V) : \sup_{0\le s \le t}\Big|(X_m^N(s),V_m^N(s))-(\overline{X}_m^N(s), \overline{V}_m^N(s))\Big|_{\infty}> N^{-\al}\right\}.
 \end{eqnarray}
\end{Definition}

  \noindent Here and in the following we use $|\cdot|_{\infty}$ as the supremum norm on $\R^{6N}$.
  \noindent Note that
  $$\mathbb{E}_0(S_{t+dt}-S_{t} \,|\,\mathcal{N}_{\al}) \le 0,$$
  since $S_t$ takes the value of one for $(X,V) \in \mathcal{N}_{\al}$.

\begin{Theorem}
 Let $f_m^N(t,x,v)$
 be a  solution of  the regularized Vlasov-Maxwell equation \eqref{VMN}. Suppose that  Assumptions 3.1 are satisfied.
 Then there exists a constant $M$  such that
  $$
  \mathbb{P}_0 \left(\sup_{0\le s \le t}\left|(X_m^N(s),V_m^N(s))-(\overline{X}_m^N(s), \overline{V}_m^N(s))\right|_{\infty}> N^{-\al} \right) \le e^{Mt}\cdot c^4 N^{-(1-\al-4\beta-16\theta)}.
  $$
\end{Theorem}



\noindent The proof of the theorem will be presented later in this section.\\

\begin{Definition}  The sets $\mathcal{N}_{\beta}$ and $ \mathcal{N}_{\gamma}$
  are characterized by
\begin{eqnarray} \label{nbeta}
  \mathcal{N}_{\beta}:=\left\{(X_m,V_m) :  \left| \Psi_m^{1,N}\left(\overline{X}_m^N(t),\overline{V}_m^N(t)\right)- \overline{\Psi}_m^{1,N}\left(\overline{X}_m^N(t), \overline{V}_m^N(t)\right) \right| _{\infty}> N^{-\beta}\right\},
  \end{eqnarray}
  \begin{eqnarray}\label{ngamma}
 \mathcal{N}_{\gamma}:=\left\{(X_m,V_m) :  \left| \Psi_m^{2,N}\left(\overline{X}_m^N(t),\overline{V}_m^N(t)\right)- \overline{\Psi}_m^{2,N}\left(\overline{X}_m^N(t), \overline{V}_m^N(t)\right) \right| _{\infty}> N^{-\gamma}\right\}.
 \end{eqnarray}
\end{Definition}

Next, we will see that the  probability  of both sets $\mathcal{N}_{\beta}$ and $ \mathcal{N}_{\gamma}$
 \eed tends to 0 as $N$ goes to infinity. We prove the following two lemmas:

 \begin{Lemma}
   There exists  a constant $M < \infty$ such that
   $$\mathbb{P}_0(\mathcal{N}_{\beta}) \le Mc^4 \cdot N^{-(1-4\beta-16\theta)}.$$
 \end{Lemma}

\textit{Proof.} First, we let the set $\mathcal{N}_{\beta}$ evolve along the
characteristics of the regularized Vlasov-Maxwell equation
$$\mathcal{N}_{\beta,t}:=\left\{(\overline{X}_m^N(t),\overline{V}_m^N(t)) :  \left|N^{\beta} \Psi_m^{2,N}\left(\overline{X}_m^N(t),\overline{V}_m^N(t)\right)-
N^{\beta}\overline{\Psi}_m^{1,N}\left(\overline{X}_m^N(t), \overline{V}_m^N(t)\right) \right| _{\infty}> 1 \right\}$$
and consider the following fact
$$\mathcal{N}_{\beta,t} \subseteq \bigoplus^{N}_{i=1}\mathcal{N}^i_{\beta,t},$$
where
\begin{eqnarray*}
	\mathcal{N}^i_{\beta,t}:=\Bigg\{ &&(\overline{x}_m^{i,N},\overline{v}_m^{i,N}) :  \Bigg|\\
	&& N^{\beta}\cdot \frac{1}{N-1}\sum_{j=1,j\neq i}^N\int_0^t (\hat{v}(v_m^{j,N}(s))\p_t+c^2\nabla_x )Y^N(t-s,x_m^{i,N}(t)-x_m^{j,N}(s))\,ds\\
	  &&-N^{\beta}\iint_{\R^6}\int_0^t(c^2\nabla+\hat{v}(v)\p_s )f^N_m(s,\overline{x}_m^{j,N}-y,v)Y^N(t-s,y)\,dsdydv \Bigg| _{\infty}> 1 \Bigg\}.
\end{eqnarray*}

\noindent We therefore get
$$\mathbb{P}_t(\mathcal{N}_{\beta,t}) \le \sum_{i=1}^N
\mathbb{P}_t(\mathcal{N}^i_{\beta,t})=N\mathbb{P}_t(\mathcal{N}^1_{\beta,t}),$$
where in the last step we  used symmetry in exchanging any two coordinates.

Using Markov inequality gives
\begin{eqnarray}
\label{betaexpectation1}
\mathbb{P}_t(\mathcal{N}^1_{\beta,t}) &\le& \mathbb{E}_t\left[ \left( N^{\beta}\cdot \frac{1}{N-1}\sum^N_{ j =2}F_1^N(t,\overline{x}_m^{1,N}-\overline{x}_m^{j,N})
-N^{\beta}\overline{F}_1^N(t,\overline{x}_m^{1,N}) \right)^4 \right]  \nonumber \\
&=& \left(\frac{N^{\beta}}{{N-1}}\right)^4 \mathbb{E}_t\left[ \left( \sum^N_{ j =2}F_1^N(t,\overline{x}_m^{1,N}-\overline{x}_m^{j,N})
-(N-1)\overline{F}_1^N(t,\overline{x}_m^{1,N})\right)^4 \right]. \nonumber \\
\end{eqnarray}

\noindent Let $ \displaystyle h_j:=F_1^N(t,\overline{x}_m^{1,N}-\overline{x}_m^{j,N})
-\overline{F}_1^N(t,\overline{x}_m^{1,N})$. Then, each term in
the expectation (\ref{betaexpectation1}) takes the form of $
\prod_{j=2}^Nh_j^{k_j}$ with $\sum_{j=1}^Nk_j=4 $, and more importantly, the expectation assumes the value
of zero whenever there exists a $j$ such that $k_j=1$. This can be easily
verified by integrating over the $j$-th variable first or, in other words, by
acknowledging the fact that $\forall \, j=2, \ldots, N$, there holds
$$ \mathbb{E}_t \left[F_1^N(t,\overline{x}_m^{1,N}-\overline{x}_m^{j,N})
-\overline{F}_1^N(t,\overline{x}_m^{1,N}) \right] =0.
$$
Then, we can simplify the estimate (\ref{betaexpectation1}) to
\begin{eqnarray*}
\mathbb{P}_t(\mathcal{N}^1_{\beta,t}) \le \left(\frac{N^{\beta}}{{N-1}}\right)^4 \mathbb{E}_t\left[\, \sum_{j=2}^N h_j^4
+\sum^N_{2 \le m<n}\begin{pmatrix} 4 \\ 2 \end{pmatrix} h_m^2h_n^2 \right].
\end{eqnarray*}

 Since 
 $$
 \|F_m^{1,N}\|_{L^{\infty}([0,T]\times \R^3)}\leq Mc^2N^{4\theta},
 $$ 
and
\begin{eqnarray*}
	\|\overline{F}_m^{1,N}(t,x)\|_{L^{\infty}([0,T]\times \R^3)} \leq cM\Big(\sup_{0\leq t\leq T}\|f^N_m(t,\cdot,\cdot)\|_{L^{\infty}(\R^3\times \R^3)}+\sup_{0\leq t\leq T}\|\nabla_xf^N_m(t,\cdot,\cdot)\|_{L^{\infty}(\R^3\times \R^3)}\Big)
\end{eqnarray*}

we thus have  for any fixed $j$
\begin{eqnarray*}
 |h_j| &\le& |F_1^N(t,\overline{x}_m^{1,N}-\overline{x}_m^{j,N})| +  |\overline{F}_1^N(t,\overline{x}_m^{1,N})| \\
 &\le& cM\left(N^{4\theta}+\sup_{0\leq t\leq T}\|f^N_m(t,\cdot,\cdot)\|_{L^{\infty}(\R^3\times \R^3)}+\sup_{0\leq t\leq T}\|\nabla_xf^N_m(t,\cdot,\cdot)\|_{L^{\infty}(\R^3\times \R^3)}\right).
\end{eqnarray*}
Therefore $|h_j|$ is bounded to any power and we obtian
$$  \mathbb{E}_t \left[ h_m^2h_n^2 \right] \le Mc^4N^{16\theta} \quad \hbox{and} \quad \mathbb{E}_t \left[ h_j^4 \right] \le Mc^4N^{16\theta}$$

and consequently
\begin{eqnarray*}
\mathbb{P}_t(\mathcal{N}^1_{\beta,t}) &\le& \left(\frac{N^{\beta}}{{N-1}}\right)^4 \cdot \left(  Mc^4 \cdot (N-1)+
Mc^4N^{16\theta}\cdot \frac{(N-1)(N-2)}{2}  \right) \\
&\le& Mc^4N^{16\theta} \cdot N^{-(2-4\beta-16\theta)}.
\end{eqnarray*}

\noindent By noticing the fact that
\begin{eqnarray*}
	\mathbb{P}_0(\mathcal{N}_{\beta})&=&\mathbb{P}_t\left(\mathcal{N}_{\beta,t}\right) \leq N\mathbb{P}_t(\mathcal{N}^1_{\beta,t})\\
&\le& N \cdot Mc^4 \cdot N^{-(2-4\beta-16\theta)}=Mc^4 \cdot N^{-(1-4\beta-16\theta)},
\end{eqnarray*}

we obtain the desired result.
\qed\\

  \noindent In fact, this result holds for any $\beta$  if we change   the 
  power in the proof to be another even number (depending on $\beta$) greater than four.  So with similar estimates we get:

  \begin{Lemma}
   There exists  a constant $M < \infty$ such that
   $$\mathbb{P}_0(\mathcal{N}_{\gamma}) \le Mc^4 \cdot N^{-(1-4\gamma-16\theta)}.$$
 \end{Lemma}

\begin{Lemma}
  Let $\mathcal{N}_{\al}$, $ \mathcal{N}_{\beta}$, $ \mathcal{N}_{\gamma}$ be defined
  as in (\ref{nalpha})-(\ref{ngamma}). Suppose that $f_m^N(t,x,v)$ is  a  solution of  the regularized Vlasov-Maxwell equation and Assumption 3.1 is satisfied. Then there exists a constant $M<\infty$ such that
\begin{eqnarray*}
  &&\Big|\Big(\hat{V}(V_m^N(t)), \Psi_m^{1,N}(t,X_m^N(t),V_m^N(t))+\Psi_m^{2,N}(t,X_m^N(t),V_m^N(t))+\Gamma_m^N(t,X_m^N(t),V_m^N(t))\Big)\\
  &&\hspace{0.5cm}-\Big(\hat{V}(\overline{V}_m^N(t)), \overline{\Psi}_m^{1,N}(t,\overline{X}_m^N(t), \overline{V}_m^N(t))+\overline{\Psi}_m^{2,N}(t,\overline{X}_m^N(t), \overline{V}_m^N(t))+\Gamma_m^{N}(t,\overline{X}_m^N(t),\overline{V}_m^N(t))\Big)
  \Big| _{\infty}\\
  &&\hspace{11cm} \le M S_t(X,V)N^{-\al}+N^{-\beta}
\end{eqnarray*}
  for all initial data $(X,V) \in (\mathcal{N}_{\al} \cup \mathcal{N}_{\beta} \cup
  \mathcal{N}_{\gamma})^c$.
\end{Lemma}

\textit{Proof.} Applying triangle inequality gives
\begin{eqnarray*}
 &&\Big|\Big(\hat{V}(V_m^N(t)), \Psi_m^{1,N}(t,X_m^N(t),V_m^N(t))+\Psi_m^{2,N}(t,X_m^N(t),V_m^N(t))+\Gamma_m^N(t,X_m^N(t),V_m^N(t))\Big)\\
  &&\hspace{0.5cm}-\Big(\hat{V}(\overline{V}_m^N(t)), \overline{\Psi}_m^{1,N}(t,\overline{X}_m^N(t), \overline{V}_m^N(t))+\overline{\Psi}_m^{2,N}(t,\overline{X}_m^N(t), \overline{V}_m^N(t))+\Gamma_m^{N}(t,\overline{X}_m^N(t),\overline{V}_m^N(t))\Big)
  \Big| _{\infty}\\
  &\le& \left|\hat{V}(V_m^N(t))-\hat{V}(\overline{V}_m^N(t))\right| _{\infty}+\left|\Psi_m^{1,N}(t,X_m^N(t),V_m^N(t))-\overline{\Psi}_m^{1,N}(t,\overline{X}_m^N(t), \overline{V}_m^N(t))\right| _{\infty}\\
  &&+\left|\Psi_m^{2,N}(t,X_m^N(t),V_m^N(t))-\overline{\Psi}_m^{2,N}(t,\overline{X}_m^N(t), \overline{V}_m^N(t))\right| _{\infty}\\
  &&+\left|\Gamma_m^{N}(t,X_m^N(t),V_m^N(t))-\Gamma_m^{N}(t,\overline{X}_m^N(t),\overline{V}_m^N(t))\right| _{\infty} \\
&\le& \left|\hat{V}(V_m^N(t))-\hat{V}(\overline{V}_m^N(t))\right|_{\infty}+\left|\Psi_m^{1,N}(t,X_m^N(t),V_m^N(t))-\Psi_m^{1,N}(t,\overline{X}_m^N(t), \overline{V}_m^N(t))\right| _{\infty}  \\
&& +\left|\Psi_m^{1,N}(t,\overline{X}_m^N(t), \overline{V}_m^N(t))-\overline{\Psi}_m^{1,N}(t,\overline{X}_m^N(t), \overline{V}_m^N(t))\right| _{\infty}\\
&&+\left|\Psi_m^{2,N}(t,X_m^N(t),V_m^N(t))-\Psi_m^{2,N}(t,\overline{X}_m^N(t), \overline{V}_m^N(t))\right| _{\infty}\\
&& +\left|\Psi_m^{2,N}(t,\overline{X}_m^N(t), \overline{V}_m^N(t))-\overline{\Psi}_m^{2,N}(t,\overline{X}_m^N(t), \overline{V}_m^N(t))\right| _{\infty}\\
 &&+ \left|\Gamma_m^{N}(t,X_m^N(t),V_m^N(t))-\Gamma_m^{N}(t,\overline{X}_m^N(t),\overline{V}_m^N(t))\right| _{\infty}\\
 &=:& |I_1|+|I_2|+|I_3|+|I_4|+|I_5|+|I_6|.
\end{eqnarray*}
Next, we estimate term by term.
\begin{itemize}
  \item Since $(X,V) \notin \mathcal{N}_{\al}$,
  $$|I_1|:=\left|\hat{V}(V_m^N(t))-\hat{V}(\overline{V}_m^N(t))\right|_{\infty} \le MS_t(X,V)N^{-\al}.$$

\item   Note,  that $F_m^{1,N}$ is
Lipschitz continuous in $x$. We denote $L$ as the  global  Lipschitz constant for all the Lipschitz continuous functions in this paper. Thus we obtain
 \begin{eqnarray}
   \label{I2}
   && \left|\frac{1}{N-1}\sum_{i \neq j}F_m^{1,N}(t, x_m^{i,N})-\frac{1}{N-1}\sum_{i \neq j}F_m^{1,N}(t,\overline{x}_m^{i,N}) \right|  \nonumber \\
   &\le& \frac{1}{N-1}\sum_{i \neq j} L\cdot  2|x_m^{i,N}-\overline{x}_m^{i,N}|.
 \end{eqnarray}
Since $(X,V) \notin \mathcal{N}_{\al}$, it follows in particular for any $1 \le i \le 
N$ that
$$ |x_m^{i,N}-\overline{x}_m^{i,N}|\le N^{-\al} .$$

So together with (\ref{I2}), we have
$$\left|\Psi_m^{1,N}(t,X_m^N(t),V_m^N(t))-\Psi_m^{1,N}(t,\overline{X}_m^N(t), \overline{V}_m^N(t))\right| _{\infty} \le 2 LN^{-\al},$$

and thus
$$|I_2|:=\left|\Psi_m^{1,N}(t,X_m^N(t),V_m^N(t))-\Psi_m^{1,N}(t,\overline{X}_m^N(t), \overline{V}_m^N(t))\right| _{\infty} \le MS_t(X,V)N^{-\al}.$$
Similarly 
$$|I_4|:=\left|\Psi_m^{2,N}(t,X_m^N(t),V_m^N(t))-\Psi_m^{2,N}(t,\overline{X}_m^N(t), \overline{V}_m^N(t))\right| _{\infty} \le MS_t(X,V)N^{-\al}.$$

  \item Since $(X,V) \notin \mathcal{N}_{\beta}$, it follows directly
 $$|I_3|:=\left|\Psi_m^{1,N}(t,\overline{X}_m^N(t), \overline{V}_m^N(t))-\overline{\Psi}_m^{1,N}(t,\overline{X}_m^N(t), \overline{V}_m^N(t))\right| _{\infty}
 \le N^{-\beta}.$$
 
   \item Since $(X,V) \notin \mathcal{N}_{\gamma}$, it follows directly
  $$|I_5|:=\left|\Psi_m^{2,N}(t,\overline{X}_m^N(t), \overline{V}_m^N(t))-\overline{\Psi}_m^{2,N}(t,\overline{X}_m^N(t), \overline{V}_m^N(t))\right| _{\infty}
 \le N^{-\beta}.$$
 
 \item Since $E_0$  and $B_0$ under Assumption 4.1(a) are Lipschitz continuous, we have for each $1 \le i \le N$, $(x_m^{i,N},v_m^{i,N}) = \big((X_m^N(t), V_m^N(t))\big)_i$, $(\overline{x}_m^{i,N},\overline{v}_m^{i,N}) =\big ((\overline{X}_m^N(t),
 \overline{V}_m^N(t))\big)_i$
and together with the fact that $(X,V) \notin \mathcal{N}_{\al}$, there holds
 $$|I_6|:=\left|\Gamma_m^{N}(t,X_m^N(t),V_m^N(t))-\Gamma_m^{N}(t,\overline{X}_m^N(t),\overline{V}_m^N(t))\right| _{\infty}\ \le LS_t(X,V)N^{-\al}.$$

\end{itemize}

 \noindent Combining all the six terms, we end up with
  \begin{eqnarray*}
  &&\Big|\Big(\hat{V}(V_m^N(t)), \Psi_m^{1,N}(t,X_m^N(t),V_m^N(t))+\Psi_m^{2,N}(t,X_m^N(t),V_m^N(t))+\Gamma_m^N(t,X_m^N(t),V_m^N(t))\Big)\\
  &&\hspace{0.5cm}-\Big(\hat{V}(\overline{V}_m^N(t)), \overline{\Psi}_m^{1,N}(t,\overline{X}_m^N(t), \overline{V}_m^N(t))+\overline{\Psi}_m^{2,N}(t,\overline{X}_m^N(t), \overline{V}_m^N(t))+\Gamma_m^{N}(t,\overline{X}_m^N(t),\overline{V}_m^N(t))\Big)
  \Big| _{\infty}\\
  &&\hspace{11cm} \le M S_t(X,V)N^{-\al}+N^{-\beta}
\end{eqnarray*}
  for all $(X,V) \in (\mathcal{N}_{\al} \cup \mathcal{N}_{\beta} \cup
  \mathcal{N}_{\gamma})^c$.
 \qed\\

\noindent Using all the Lemmas above we can now prove Theorem 3.1:

\subsubsection*{Proof of Theorem 3.1}
From the definition of the Newtonian flow (\ref{VMNPF}) and the characteristics of the Vlasov equation
(\ref{VMNF}), we know  that 
\begin{eqnarray*}
 && (X^N_{m}(t+dt), V^N_{m}(t+dt)) \\
 &=& (X_m^N(t), V_m^N(t))\\
 &&+\Big(\hat{V}(V_m^N(t)), \Psi_m^{1,N}(t,X_m^N(t),V_m^N(t))+\Psi_m^{2,N}(t,X_m^N(t), V_m^N(t))+\Gamma_m^{N}(t, X_m^N(t), V_m^N(t))\Big)dt+o(dt),
 \end{eqnarray*}
 and
 \begin{eqnarray*}
  && (\overline{X}^N_{m}(t+dt), \overline{V}^N_{m}(t+dt)) \\
  &=& (\overline{X}_m^N(t), \overline{V}_m^N(t))\\
  &&+\Big(\hat{V}(\overline{V}_m^N(t)), \overline{\Psi}_m^{1,N}(t,\overline{X}_m^N(t))+\overline{\Psi}_m^{2,N}(t,\overline{X}_m^N(t), \overline{V}_m^N(t))+\Gamma_m^{N}(t,\overline{X}_m^N(t), \overline{V}_m^N(t))\Big)dt+o(dt).
\end{eqnarray*}
Thus
\begin{eqnarray*}
	 &&\left|  (X^N_{m}(t+dt), V^N_{m}(t+dt)) -  (\overline{X}^N_{m}(t+dt), \overline{V}^N_{m}(t+dt))\right|_{\infty} \le \left|(X_m^N(t), V_m^N(t))- (\overline{X}_m^N(t), \overline{V}_m^N(t))\right|_{\infty} \\
	 && +\Big|\Big(\hat{V}(V_m^N(t)), \Psi_m^{1,N}(t,X_m^N(t),V_m^N(t))+\Psi_m^{2,N}(t,X_m^N(t), V_m^N(t))+\Gamma_m^{N}(t, X_m^N(t), V_m^N(t))\Big)\\
	 &&\hspace{0.5cm}-\Big(\hat{V}(\overline{V}_m^N(t)), \overline{\Psi}_m^{1,N}(t,\overline{X}_m^N(t))+\overline{\Psi}_m^{2,N}(t,\overline{X}_m^N(t), \overline{V}_m^N(t))+\Gamma_m^{N}(t,\overline{X}_m^N(t), \overline{V}_m^N(t))\Big)\Big|_{\infty}dt+o(dt),
\end{eqnarray*}
i.e.,
\begin{eqnarray*}
&& S_{t+dt}-S_t \\
&\le& \Big|\Big(\hat{V}(V_m^N(t)), \Psi_m^{1,N}(t,X_m^N(t),V_m^N(t))+\Psi_m^{2,N}(t,X_m^N(t), V_m^N(t))+\Gamma_m^{N}(t, X_m^N(t), V_m^N(t))\Big)\\
	 &&\hspace{0.5cm}-\Big(\hat{V}(\overline{V}_m^N(t)), \overline{\Psi}_m^{1,N}(t,\overline{X}_m^N(t))+\overline{\Psi}_m^{2,N}(t,\overline{X}_m^N(t), \overline{V}_m^N(t))+\Gamma_m^{N}(t,\overline{X}_m^N(t), \overline{V}_m^N(t))\Big)\Big|_{\infty}dt+o(dt),
\end{eqnarray*}

\noindent Taking the expectation over both sides yields
\begin{eqnarray*}
 && \mathbb{E}_0\left[\,S_{t+dt}-S_t\right]  \\
 &=&  \mathbb{E}_0\left[S_{t+dt}-S_t \,|\, \mathcal{N}_{\al}\,\right]
  +\mathbb{E}_0\left[\,S_{t+dt}-S_t \,|\, \mathcal{N}_{\al}^c\,\right]
  \\
  &\le&  \mathbb{E}_0\left[\,S_{t+dt}-S_t \,|\, (\mathcal{N}_{\beta} \cup \mathcal{N}_{\gamma})
  \setminus \mathcal{N}_{\al}\,\right]
  +\mathbb{E}_0\left[\,S_{t+dt}-S_t \,|\, (\mathcal{N}_{\al} \cup \mathcal{N}_{\beta}
 \cup \mathcal{N}_{\gamma})^c\,\right] \\
  & \le &  \mathbb{E}_0 \left[ \, \left|\hat{V}(V_m^N(t))-\hat{V}(\overline{V}_m^N(t)) \right|_{\infty}\,\Big|\,
  (\mathcal{N}_{\beta} \cup \mathcal{N}_{\gamma}) \setminus \mathcal{N}_{\al}\, \right]N^{\al}dt \\
  && +\, \mathbb{E}_0 \left[ \, \left| \Psi_m^{1,N}(t,X_m^N(t),V_m^N(t))-\overline{\Psi}_m^{1,N}(t,\overline{X}_m^N(t)) \right|_{\infty}\,\Big|\, (\mathcal{N}_{\beta} \cup \mathcal{N}_{\gamma}) \setminus \mathcal{N}_{\al}\, \right]N^{\al}dt \\
  && +\, \mathbb{E}_0 \left[ \, \left| \Psi_m^{2,N}(t,X_m^N(t),V_m^N(t))-\overline{\Psi}_m^{2,N}(t,\overline{X}_m^N(t),\overline{V}_m^N(t)) \right|_{\infty}\,\Big|\, (\mathcal{N}_{\beta} \cup \mathcal{N}_{\gamma}) \setminus \mathcal{N}_{\al}\, \right]N^{\al}dt \\
  && +\, \mathbb{E}_0 \left[ \, \left|\Gamma_m^N(t,X_m^N(t),V_m^N(t))-\Gamma_m^N(t,\overline{X}_m^N(t),\overline{V}_m^N(t))\right|_{\infty}\,\Big|\,
(\mathcal{N}_{\beta} \cup \mathcal{N}_{\gamma}) \setminus \mathcal{N}_{\al}\, \right]N^{\al}dt\\
&& +\, \mathbb{E}_0\left[\,S_{t+dt}-S_t \,|\, (\mathcal{N}_{\al} \cup \mathcal{N}_{\beta}
 \cup \mathcal{N}_{\gamma})^c\,\right]+o(dt)\\
 &=:& J_1+J_2+J_3+J_4+J_5+o(dt),
   \end{eqnarray*}
   where in the second step we use $\mathbb{E}_0(S_{t+dt}-S_{t} \,|\,\mathcal{N}_{\al}) \le 0$  and decompose
   the set $\mathcal{N}^c_{\al} \, $ into
 $(\mathcal{N}_{\beta} \cup \mathcal{N}_{\gamma}) \setminus \mathcal{N}_{\al} \,$ and
  $ (\mathcal{N}_{\al} \cup \mathcal{N}_{\beta} \cup
  \mathcal{N}_{\gamma})^c\,$.

\noindent  Since $(X,V) \notin \mathcal{N}_{\al}$, it follows
 \begin{eqnarray*}
J_1 &=& \mathbb{E}_0 \left[ \, \left|\hat{V}(V_m^N(t))-\hat{V}(\overline{V}_m^N(t)) \right|_{\infty}\,\Big|\,
  (\mathcal{N}_{\beta} \cup \mathcal{N}_{\gamma}) \setminus \mathcal{N}_{\al}\, \right]N^{\al}dt 
  \\
&\le& L
\big(\mathbb{P}_0(\mathcal{N}_{\beta})+\mathbb{P}_0(\mathcal{N}_{\gamma})\big)dt.
  \end{eqnarray*}

\noindent Due to the definition of $\Psi_m^{1,N}$, $\overline{\Psi}_m^{1,N}$, $\Psi_m^{2,N}$, $\overline{\Psi}_m^{2,N}$ and $\Gamma_m^{N}$, we obtain
$$ J_2+J_3+J_4 \le  M \big(\mathbb{P}_0(\mathcal{N}_{\beta})+
\mathbb{P}_0(\mathcal{N}_{\gamma})\big)N^{\al}dt.$$

\noindent Thanks to Lemma 4.1 and Lemma 4.2, we get
\begin{eqnarray*}
J_1+J_2+J_3+J_4 &\leq & M \big(\mathbb{P}_0(\mathcal{N}_{\beta})+\mathbb{P}_0(\mathcal{N}_{\gamma})\big)N^{\al}dt
\\
&\le& M c^4 \cdot N^{-(1-4\beta-16\theta)} N^{\al} dt.
 \end{eqnarray*}
On the other hand, Lemma 3.4 states that
   \begin{eqnarray*}
     J_5 &=& \mathbb{E}_0\left[\,S_{t+dt}-S_t \,|\, (\mathcal{N}_{\al} \cup \mathcal{N}_{\beta}
 \cup \mathcal{N}_{\gamma})^c\,\right] \\
 &\le& (M \cdot \mathbb{E}_0\left[S_t\right]N^{-\al}+N^{-\beta})\cdot N^{\al}dt +o(dt)\\
 &=& M \cdot \mathbb{E}_0\left[S_t\right]dt+N^{\al-\beta}dt+o(dt).
   \end{eqnarray*}
Therefore, we can determine the estimate
\begin{eqnarray*}
\mathbb{E}_0\left[S_{t+dt}\right]-\mathbb{E}_0[S_{t}] &\le& \mathbb{E}_0\left[S_{t+dt}-S_t \right] \\
&\le& M \cdot \mathbb{E}_0\left[S_t\right]dt
+M \cdot c^4 \cdot N^{-(1-\al-4\beta-16\theta)}dt+o(dt).
\end{eqnarray*}

Equivalently, we have
$$ \frac{d}{dt}\,\mathbb{E}_0[S_{t}] \le M \cdot \mathbb{E}_0[S_{t}] +M \cdot c^4 \cdot N^{-(1-\al-4\beta-16\theta)}.$$

\noindent Gronwall's inequality yields
$$ \mathbb{E}_0\left[S_t\right] \le e^{Mt}\cdot c^4 \cdot N^{-(1-\al-4\beta-16\theta)}.$$
The proof is completed by the following Markov inequality
$$\mathbb{P}_0 \left(\sup_{0\le s \le t}\Big|(X_m^N(s),V_m^N(s))-(\overline{X}_m^N(s), \overline{V}_m^N(s))\Big|_{\infty}> N^{-\al} \right) =
\mathbb{P}_0 (S_t=1) \le  \mathbb{E}_0\left[S_t\right]. $$
\qed


\subsection{Estimates for the Non-relativistic Limit}

Due to the key estimate \eqref{MPestimate}, it is easy to repeat the whole procedure in the previous subsection to obtian 

\begin{Theorem}
  Let $f^N_m(t,x,v)$ and $f_p^N(t,x,v)$ be the solutions to the regularized Vlasov-Maxwell equation \eqref{VMN} and \eqref{VPN} respectively with the same
  initial data $f_0$. Suppose that Assumptions 3.1 are satisfied.  Then 
  $$\mathbb{P}_0 \left(\sup_{0\le s \le t}\Big|(\overline{X}_m^N(s),\overline{V}_m^N(s))-(\overline{X}_p^N(s), \overline{V}_p^N(s))\Big|_{\infty}> N^{-\al} \right) \leq e^{Mt}\frac{M}{c}.$$
\end{Theorem}

\begin{Remark}
	We point out that the proof is straightforward when we use the flows of \eqref{VMNF} and \eqref{VPNF}.
\end{Remark}

\subsection{Combined Limit}

Now with all the estimates we achieved above, we take $c=N^{\eta}$, $ \eta \in \left(0, \frac{1-\al-4\beta-16\theta}{4}\right)$.  
\begin{Theorem}
  Let $f^N_m(t,x,v)$ and $f_p^N(t,x,v)$ be the solutions to the regularized Vlasov-Maxwell equation \eqref{VMN} and \eqref{VPN} respectively with the same
  initial data $f_0$. Suppose that Assumption 3.1 is satisfied.  Then
  $$\lim_{N \to \infty , c \to \infty}\mathbb{P}_0 \left(\sup_{0\le s \le t}\Big|(X_m^N(s),V_m^N(s))-(\overline{X}_p^N(s), \overline{V}_p^N(s))\Big|_{\infty}> N^{-\al} \right) =0.$$
\end{Theorem}

\section{Summary}
 In this paper we compared  the time evolution of the one particle density of the $N$-particle Vlasov-Maxwell system with the Vlasov-Poisson equation. We showed closeness of both time evolutions for $N$ and $c$ being large enough. 

\section{Appendix: Proof of Theorem 2.1 }
\begin{proof} 
	\begin{enumerate}
		\item Using the same method as Kurth, R. in \cite{Kurth52}, it is easy to prove that (VPN) \ed has \eed a unique $C^1$-solution $(f^{N}_p, E^{N}_p)$ on the time interval $[0,\overline{T}>0)$.
		\item The proof of existence  of solutions  of (VMN) is similar to
Glassey, R., Strauss, W \cite{GS86}, while the proof of
existence  of functions $q(t)$ and $F(t)$ with the respective properties follows the ideas of Jack
Schaeffer as given in \cite{Schaeffer86}. Therefore we omit the proof in this manuscript. 

\item 
Next we  prove the third part of the theorem.
%
%
Similar to (A13) and (A14) in  the Appendix of \cite{Schaeffer86},  we use the convenient notation $\displaystyle \nu=\frac{y-x}{|y-x|}, x,y\in{\mathbb{R}^{3}}$. 
  We  obtain
\begin{eqnarray*}
\label{SMWN3}
 &&E^{N}_{m}(t,x)\\
 &=& \mathbb{E}_{0}-\frac{1}{4\pi ct}\iiint_{\R^9}\,dvdpdz\int_{|x-y|=ct}\,dS_y {\chi}^{N}(p){\chi}^{N}(z)f_{0}(y-p-z,v) \frac{ \nu-c^{-2}\hat{v}\cdot{\nu}\hat{v}}{(1+c^{-1}\hat{v}\cdot{\nu})}\\
&&-\frac{1}{4\pi}\iiint_{\R^9}\,dvdpdz\int_{|x-y|<ct}\,dy\\
&&\hspace{2cm}{\chi}^{N}(p){\chi}^{N}(z)\frac{f^{N}_{m}(t-c^{-1}|x-y|,y-p-z,v)}{|x-y|^{2}} \frac{(1-c^{-2}|\hat{v}|^2)(\nu+c^{-1}\hat{v})}{(1+c^{-1}\hat{v}\cdot{\nu})^{2}}\\
&&-\frac{1}{4\pi c^2}\iiint_{\R^9}\,dvdpdz\int_{|x-y|<ct}\,dy{\chi}^{N}(p){\chi}^{N}(z)\frac{f^{N}_{m}(t-c^{-1}|x-y|, y-p-z, v)}{|x-y|(1+c^{-1}\hat{v}\cdot{\nu})^{2}(1+c^{-2}|v|^2)^{\frac{1}{2}}}\\
&&\times \Big[(1+c^{-1}\hat{v}\cdot{\nu})(E^{N}_{m}+c^{-1}\hat{v}\times{B^{N}_{m}})
+c^{-2}(\hat{v}\cdot{\nu}\nu-\hat{v})\hat{v}\cdot{E^{N}_{m}}\\
&& \hspace{5cm} -(\nu+c^{-1}\hat{v})\nu\cdot
(E^{N}_{m}+c^{-1}\hat{v}\times{B^{N}_{m}})\Big]\Big|_{(t-c^{-1}|x-y|,y-p-z)}\\
&=&\mathbb{E}_{0}-\mathbb{E}_{1}-\mathbb{E}_{2}-\mathbb{E}_{3},
\end{eqnarray*}

and
\begin{eqnarray*}
 &&B^{N}_{m} \\
 &=& \mathbb{B}_{0}+\frac{1}{4\pi ct}\iiint_{\R^9}\,dvdpdz\int_{|x-y|=ct}\,dS_y{\chi}^{N}(p){\chi}^{N}(z)f_{0}(y-p-z,v)
 \frac{(\nu\times c^{-1}\hat{v})}{(1+c^{-1}\hat{v}\cdot{\nu})}\\
&&+\frac{1}{4\pi
c}\iiint_{\R^9}\,dvdpdz\int_{|x-y|<ct}\,dy\\
&&\hspace{2cm} {\chi}^{N}(p){\chi}^{N}(z)\frac{f^{N}_{m}(t-c^{-1}|x-y|,y-p-z,v)}{|x-y|^{2}}
 \frac{(1-c^{-2}|\hat{v}|^2)(\nu\times \hat{v})}{(1+c^{-1}\hat{v}\cdot{\nu})^{2}}
\\
&&+\frac{1}{4\pi c^2}\iiint_{\R^9}\,dvdpdz\int_{|x-y|<ct}\,dy{\chi}^{N}(p){\chi}^{N}(z)\frac{f^{N}_{m}(t-c^{-1}|x-y|, y-p-z, v)}{|x-y|(1+c^{-1}\hat{v}\cdot{\nu})^{2}(1+c^{-2}|v|^2)^{\frac{1}{2}}}\\
&&\times \Big[(1+c^{-1}\hat{v}\cdot{\nu})\nu\times(E^{N}_{m}+c^{-1}\hat{v}\times{B^{N}_{m}})
 \\
&&\hspace{4cm} -c^{-2}(\nu\times\hat{v})(\hat{v}+c\nu)\cdot
(E^{N}_{m}+c^{-1}\hat{v}\times{B^{N}_{m}})\Big]\Big|_{(t-c^{-1}|x-y|,y-p-z)}\\
&=& \mathbb{B}_{0}+ \mathbb{B}_{1}+ \mathbb{B}_{2}+ \mathbb{B}_{3},
\end{eqnarray*}
where $\Big|_{(t-c^{-1}|x-y|,y-p-z)}$ means $E^{N}_{m}(t-c^{-1}|x-y|,y-p-z)$ and $B^{N}_{m}(t-c^{-1}|x-y|,y-p-z)$.
 In order to prove Theorem 2.1, we  note  that the core of
the  proof consists  in  comparing the integral representation of
$(E^{N}_{m}, B^{N}_{m})$ given above with  the one  of $ E^{N}_{p}$
given in (VPN) that is
$$E^{N}_{p}(t,x)=\frac{1}{4\pi}\iiiint_{\R^{12}}\,dvdydpdz{\chi}^{N}(p){\chi}^{N}(z)f^{N}_{p}(t,y-p-z)\frac{x-y}{|x-y|^{3}}.$$
To obtain uniform convergence, we will  thoroughly calculate  $E^{N}_{m}$ and $
B^{N}_{m}$. First, we consider $E^{N}_{m}$.

\begin{Lemma}
(\cite{Schaeffer86}, Lemma 1) Let $g$ be a
continuous function of compact support on $\mathbb{R}^3,$ then there
exists a constant $M>0$ such that
$$r\int\limits_{|\omega|=1}|g(x+r\omega)|\,d\omega \leq M.$$
for all $r>0$.
 \end{Lemma}
 Note that for $|v|\leq q(t),$ with $q(t)\geq 1,$
 $$
 |\hat{v}|\leq \frac{q(t)}{(1+c^{-2} |v|^2)^{\frac{1}{2}} }\leq  q(t),
 $$
 and
 $$
 \frac{1}{1+c^{-1}\hat{v}\cdot \nu}\leq 2c^{-2}(c^2+q^2(t))\leq 4q^2(t).
 $$
From the proposition and the above two inequalities, we get $\forall x \in \R^3$, $t \in [0,T]$
\begin{eqnarray*}
&&\Big|\frac{1}{4\pi
ct}\iiint_{\R^9} \,dvdpdz \int_{|x-y|=ct}\,dS_y{\chi}^{N}(p){\chi}^{N}(z)f_{0}(y-p-z,v)\frac{
c^{-2}\hat{v}\cdot{\nu}\hat{v}}{(1+c^{-1}\hat{v}\cdot{\nu})}\Big|\\
&\leq& \frac{M}{
c^2}tq^4(t)\iiint_{\R^9}\,dvdpdz {\chi}^{N}(p){\chi}^{N}(z)\int_{|\omega|=1}\,d\omega ctf_{0}(x-p-z+ct\omega,v)\\
&\leq & Mq^4(t)c^{-2}=O(c^{-2})
\end{eqnarray*}

and
\begin{eqnarray*}
&&\Big|\frac{1}{4\pi
ct}\iiint_{\R^9} \,dvdpdz \int_{|x-y|=ct}\,dS_y{\chi}^{N}(p){\chi}^{N}(z)f_{0}(y-p-z,v)\frac{
c^{-1}\hat{v}\cdot{\nu}\nu}{(1+c^{-1}\hat{v}\cdot{\nu})}\Big|\\
&\leq& \frac{M}{
c}tq^3(t)\iiint_{\R^9} \,dvdpdz {\chi}^{N}(p){\chi}^{N}(z)\int_{|\omega|=1}\,d\omega ctf_{0}(x-p-z+ct\omega,v)\\
&\leq& Mq^3(t)c^{-2}=O(c^{-1}).
\end{eqnarray*}
Hence $\forall x \in \R^3$, $t\in [0,T]$
$$
\mathbb{E}_{1}(t,x)= \frac{1}{4\pi ct}\iiint_{\R^9} \,dvdpdz{\int}_{|x-y|=ct}\,dS_y{{\chi}^{N}(p){\chi}^{N}(z)f_{0}(y-p-z,v)\nu} +O(c^{-1}).
$$ 
As
\begin{eqnarray*}
 &&\Big|\frac{1}{4\pi}\iiint_{\R^9} \,dvdpdz{\int}_{|x-y|<ct}\,dy\\
 && \hspace{3cm} {\chi}^{N}(p){\chi}^{N}(z)\frac{f^{N}_{m}(t-c^{-1}|x-y|,y-p-z,v)}{|x-y|^{2}} \frac{|\hat{v}|^2(\nu+c^{-1}\hat{v})}{(1+c^{-1}\hat{v}\cdot{\nu})^{2}c^2}\Big|\\
&\leq & \frac{1}{4\pi c^2}{\int}_{|y|<P_0+tq(t)}{\int}_{|v|<q(t)} (4q^2(t))^2q^2(t)(1+c^{-1}q(t))\frac{\|f_0\|_{L^{\infty}(\R^3\times \R^3)}}{|x-y|^{2}} \,dvdy\\
&\leq & \frac{M}{c^{2}}q^6(t)(1+c^{-1}q(t))q^3(t){\int}_{|y|<P_0+tq(t)}\frac{1}{|x-y|^{2}}\,dy\\
&\leq & \frac{M}{c^{2}},
\end{eqnarray*}
where we have in the last step used the fact that
\begin{eqnarray}
	\label{bdd1}
	\sup_x {\int}_{|y|<P_0+tq(t)}\frac{1}{|x-y|^{2}}\,dy <M(P_0, q(t)).
\end{eqnarray}

In the same way, we obtain
\begin{eqnarray*}
	&&\Big|\frac{1}{4\pi}\iiint_{\R^9} \,dvdpdz \int_{|x-y|<ct}\,dy\\
	&&\hspace{2cm}{\chi}^{N}(p){\chi}^{N}(z)\frac{f^{N}_{m}(t-c^{-1}|x-y|,y-p-z,v)}{|x-y|^{2}}
\frac{\hat{v}}{(1+c^{-1}\hat{v}\cdot{\nu})^{2}c}\Big|
\leq  \frac{M}{c},
\end{eqnarray*}
so we have
\begin{eqnarray*}
&&\mathbb{E}_{2}(t,x)\\
&=&\frac{1}{4\pi}\iiint_{\R^9} \,dvdpdz \int_{|x-y|<ct}\,dy\\
&& \hspace{2cm} {\chi}^{N}(p){\chi}^{N}(z)\frac{f^{N}_{m}(t-c^{-1}|x-y|,y-p-z,v)}{|x-y|^{2}}
\frac{\nu}{(1+c^{-1}\hat{v}\cdot{\nu})^{2}}+O(c^{-1})\\
&=&\frac{1}{4\pi}\iiint_{\R^9} \,dvdpdz \int_{|x-y|<ct}\,dy\\
&& \hspace{2cm} {\chi}^{N}(p){\chi}^{N}(z)\frac{f^{N}_{m}(t-c^{-1}|x-y|,y-p-z,v)}{|x-y|^{2}}
\nu+O(c^{-1}),
\end{eqnarray*}
where the following estimate has been used
$$\Big|\frac{1}{(1+c^{-1}\hat{v}\cdot{\nu})^{2}}-1\Big|=\frac{|2c^{-1}\hat{v}\cdot{\nu}+c^{-2}(\hat{v}\cdot{\nu})^2|}{(1+c^{-1}\hat{v}\cdot{\nu})^{2}}
\leq \frac{M}{c}q^4(t)\left(q(t)+c^{-1}q^2(t)\right)\leq \frac{M}{c}.$$ 
Recalling
Theorem 2.1 and $|\hat{v}|<c$, we get
\begin{eqnarray*}
 |\mathbb{E}_{3}| &\leq &
\frac{1}{4\pi c^2}\iiint_{\R^9} \,dvdpdz \int_{|x-y|<ct}\,dy \\
&&\hspace{1cm}\left(4q^2(t)\right)^2{\chi}^{N}(p){\chi}^{N}(z) 6H(t-c^{-1}|x-y|)\frac{f^{N}_{m}(t-c^{-1}|x-y|,y-p-z,v)}{|x-y|}\\
&\leq&
\frac{M}{c^2}{\int}_{|y|<P_0+tq(t)}\frac{1}{|x-y|}\,dy \int_{|v|<q(t)}\|f_0\|_{L^\infty (\R^3 \times \R^3)}\,dv\leq
\frac{M}{c}.
\end{eqnarray*}

\begin{Lemma} (\cite{Schaeffer86}, Lemma 2)
Let $g\in C^2(\mathbb{R}^3).$ Assume that $\Delta{g}$ has  
 compact support  for $c>0$ and $t\geq 0,$
 $$\partial_t \left(t \int_{|\omega|=1}{g(x+ct\omega)}\,d\omega\right)=-\int_{|x-y|>ct}\frac{\Delta g(y)}{|x-y|}\,dy.$$
\end{Lemma}
 Now using this lemma, we estimate $\mathbb{E}_{0}$. We know
\begin{eqnarray*}
	\mathbb{E}_{0} &=& \partial_{t}\int_{|\omega|=1}{\frac{t}{4\pi} E_{m}^N(0,x+ct\omega)\,d\omega}+\frac{t}{4\pi}\iint_{\R^6} \,dpdz \int_{|\omega|=1}\,d\omega\\
	&& {\chi}^{N}(p){\chi}^{N}(z)\left(c\nabla\times B_0(x-p-z+ct\omega)-\int_{\R^3}{\hat{v}f_0(x-p-z+ct\omega,v)\,dv}\right).
\end{eqnarray*}

From Lemma 3.1, we get
$$\frac{t}{4\pi}\left|\iint_{\R^6} \,dpdz{\int_{|\omega|=1}}\,d\omega
 { {\chi}^{N}(p){\chi}^{N}(z)(c\nabla\times B_0(x-p-z+ct\omega)}\right|\leq \frac{M}{c}$$
 and by Lemma 3.2, we obtain
\begin{equation*}
\begin{split} &\frac{t}{4\pi}\left|\iint_{\R^6} \,dpdz{\int_{|\omega|=1}}\,d\omega
 { {\chi}^{N}(p){\chi}^{N}(z)\int_{\R^3}{\hat{v}f_0{(x-p-z+ct\omega,v)}}\,dv
 }\right|\\
& =\frac{1}{4\pi c}\left|\iint_{\R^6} \,dpdz {\int_{|\omega|=1}}\,d\omega
 { {\chi}^{N}(p){\chi}^{N}(z)\int_{\R^3}{\hat{v}ctf_0(x-p-z+ct\omega,v)}
 dv}\right|\leq \frac{M}{c},
\end{split}
\end{equation*}
thus
$$\mathbb{E}_{0}=\partial_{t}\int_{|\omega|=1}{\frac{t}{4\pi}
E^{N}_{m}(0,x+ct\omega)\,d\omega}+O(c^{-1}).$$ 
Now, in order to further calculate
$\mathbb{E}_{0}$, we set
 $$g(x):=\frac{1}{4\pi}\iiiint_{\R^{12}}\,dvdydpdz
{\chi}^{N}(p){\chi}^{N}(z)\frac{f_0(y-p-z,v)}{|x-y|}.$$ 
Note that $\nabla
g(x)=-E_m^N(0,x)$ and $\displaystyle \Delta {g(x)}=\iiint_{\R^{9}}\,dvdpdz{\chi}^{N}(p){\chi}^{N}(z)f_0(x-p-z,v)$.   
Using Lemma 3.2, we  get 
\begin{eqnarray*}
&& \partial_{t}\int_{|\omega|=1}{\frac{t}{4\pi}
E_{m}^N(0,x+ct\omega)\,d\omega}\\
&=&-\partial_{t}\int_{|\omega|=1}{\frac{t}{4\pi}
\nabla g(x+ct\omega)\,d\omega}\\
&=&-\frac{1}{4\pi}\nabla \int_{|x-y|>ct}\frac{\Delta
g(y)}{|x-y|}\,dy \\
&=&-\frac{1}{4\pi}\nabla \iiint_{\R^9}\,dvdpdz\int_{|x-y|>ct}\,dy{\chi}^{N}(p){\chi}^{N}(z)\frac{f_0(y-p-z,v)}{|x-y|}\\
&=&-\frac{1}{4\pi}\iiint_{\R^9}\,dvdpdz \int_{|x-y|>ct}\,dy{\chi}^{N}(p){\chi}^{N}(z)\frac{\nabla_y
f_0(y-p-z,v)}{|x-y|}.
\end{eqnarray*}

Recall that $f_0$ has compact support, so by the divergence theorem,
we have
\begin{eqnarray*}
&&-\iiint_{\R^9}\,dvdpdz\int_{|x-y|=ct}\,dS_y{\chi}^{N}(p){\chi}^{N}(z)
\frac{f_0(y-p-z,v)\nu}{|x-y|} \\
&=&\iiint_{\R^9}\,dvdpdz\int_{|x-y|>ct}\,dy\nabla_y\left({\chi}^{N}(p){\chi}^{N}(z)
\frac{f_0(y-p-z,v)}{|x-y|}\right) \\
&=&\iiint_{\R^9}\,dvdpdz\int_{|x-y|>ct}\,dy{\chi}^{N}(p){\chi}^{N}(z)
\frac{\nabla_y f_0(y-p-z,v)}{|x-y|} \\
&&-\iiint_{\R^9}\,dvdpdz\int_{|x-y|>ct}\,dy{\chi}^{N}(p){\chi}^{N}(z)
 \frac{f_0(y-p-z,v)\nu}{|x-y|^{2} }
\end{eqnarray*}
Hence
\begin{eqnarray*}
\mathbb{E}_{0}&=&\frac{1}{4\pi}\iiint_{\R^9}\,dvdpdz\int_{|x-y|=ct}\,dS_y{\chi}^{N}(p){\chi}^{N}(z)
\frac{f_0(y-p-z,v)\nu}{|x-y|}\\
&&-\frac{1}{4\pi}\iiint_{\R^9}\,dvdpdz\int_{|x-y|>ct}\,dy{\chi}^{N}(p){\chi}^{N}(z)
\frac{ f_0(y-p-z,v) \nu}{|x-y|^{2}} +O(c^{-1}).
\end{eqnarray*}
Therefore
\begin{eqnarray*}
&&E^{N}_{m}(t,x)\\
&=&-\frac{1}{4\pi}\iiint_{\R^9}\,dvdpdz\int_{|x-y|<ct}\,dy{\chi}^{N}(p){\chi}^{N}(z)\frac{f^{N}_{m}(t-c^{-1}|x-y|,y-p-z,v)
\nu}{|x-y|^{2}}\\
&&-\frac{1}{4\pi}\iiint_{\R^9}\,dvdpdz\int_{|x-y|>ct}\,dy{\chi}^{N}(p){\chi}^{N}(z)
 \frac{f_0(y-p-z,v) \nu}{|x-y|^{2}} +O(c^{-1})\\
&=&-\frac{1}{4\pi}\iiiint_{\R^{12}}\,dvdpdzdy \\
&& \hspace{1cm} {\chi}^{N}(p){\chi}^{N}(z)\frac{f^{N}_{m}(\max\{0,t-c^{-1}|x-y|\},y-p-z,v)
\nu} {|x-y|^{2}}+O(c^{-1}).
\end{eqnarray*}
From the representation of $E^{N}_{p}(t,x)$ from (VPN), we have
\begin{eqnarray*}
&&| E^{N}_{m}(t,x)-E^{N}_{p}(t,x)|\\
&=&\frac{1}{4\pi}\Big|\iiiint_{\R^{12}}\,dvdpdzdy {\chi}^{N}(p){\chi}^{N}(z)\frac{\nu}{|x-y|^{2}}\\
&&\hspace{0.8cm} \times \left(f^{N}_{m}(\max\{0,t-c^{-1}|x-y|\},y-p-z,v)
-f^{N}_{p}(t,y-p-z,v) \right)\Big| +O(c^{-1})\\
&\leq&
\frac{M}{c}+\frac{1}{4\pi}\iiiint_{\R^{12}}\,dvdpdzdy\\
&& \hspace{2cm}{\chi}^{N}(p){\chi}^{N}(z)\frac{|f^{N}_{m}-f^{N}_{p}|(\max\{0,t-c^{-1}|x-y|\},y-p-z,v)}{|x-y|^{2}}
\\
&&+\frac{1}{4\pi}\iiiint_{\R^{12}}\,dvdpdzdy\\
&&\hspace{1cm} {\chi}^{N}(p){\chi}^{N}(z)\frac{|f^{N}_{p}(\max\{0,t-c^{-1}|x-y|\},y-p-z,v)
-f^{N}_{p}(t,y-p-z,v)|}{|x-y|^{2}}.
\end{eqnarray*}
Recall that $(f^{N}_{p},E^{N}_{p})$ is a  $C^1$-solution of (VPN). Now
since $E^{N}_{p}$ is $C^1$ and $f_0$ has compact support, it follows
that$$q^{N}_{p}=\sup\{|v|:\exists \, x\in\mathbb{R}^3,
~\tau\in[0,t]~\hbox{s.t.}~f^{N}_{p}(\tau,x,v)\neq0\}$$ is finite on
$[0,T].$ Also $\partial_{t}f^{N}_{p}$ is bounded on
$[0,T]\times \mathbb{R}^6$. Let 
$$
Q:=\max\{q(T),q^{N}_{p}(T)\}
$$ 
and
$$
G(t):=\sup\left\{|f^{N}_{m}(\tau,x,v)-f^{N}_{p}(\tau,x,v)|:
x\in\mathbb{R}^3, v\in\mathbb{R}^3\, \hbox{and} \, \tau \in[0,t]\right\}.
$$ Then
\begin{eqnarray*}
\label{EMP}
&&| E^{N}_{m}(t,x)-E^{N}_{p}(t,x)|\\
&\leq&
\int_{|y|<P_0+TQ}\int_{|v|<Q}\frac{G(\max\{0,t-c^{-1}|x-y|\})}{|x-y|^{2} } \,dvdy \\
&&+\iiiint_{\R^{12}}\,dvdpdzdy\\
&&\hspace{2cm}{\chi}^{N}(p){\chi}^{N}(z)\frac{1}{|x-y|^{2}}\int^{t}_{\max\{0,t-c^{-1}|x-y|\}}|\partial_{t}
f^{N}_{p}(\tau,y-p-z,v)|\,d\tau  +
\frac{M}{c}\\
&\leq&
G(t)MQ^3\int_{|y|<P_0+TQ}\frac{1}{|x-y|^{2}}\,dy +MQ^3\int_{|y|<P_0+TQ}\frac{c^{-1}|x-y|}{|x-y|^{2}}\,dy+
\frac{M}{c}\\
&\leq& MG(t)+\frac{M}{c},
\end{eqnarray*}
where we have used \eqref{bdd1}.
Now we begin to estimate $B^{N}_{m}$.  
 Using Lemma 3.2, we get for the first term $\mathbb{B}_{1}$
\begin{equation*}
\begin{split}
 |\mathbb{B}_{1}|=& \frac{1}{4\pi ct}\left|\iiint_{\R^9}\,dvdpdz\int_{|x-y|=ct}\,dS_y{\chi}^{N}(p){\chi}^{N}(z)f_{0}(y-p-z,v)\frac{
 (\nu\times c^{-1}\hat{v})}{(1+c^{-1}\hat{v}\cdot{\nu})}\right|\\
 \leq&\frac{1}{4\pi }\iiint_{\R^9}\,dvdpdz {{\chi}^{N}(p){\chi}^{N}(z)4q^{2}(t)c^{-1}q(t)\int_{|\omega|=1}ctf_{0}(x-p-z+ct\omega,v)\,d\omega}
 \\
\leq&
\frac{M}{c}\iint_{\R^6}{{\chi}^{N}(p){\chi}^{N}(z)}\,dpdz=\frac{M}{c}.
 \end{split}
\end{equation*}
 Secondly we look into $\mathbb{B}_{2}$.
\begin{eqnarray*}
|\mathbb{B}_{2}|&=&\frac{1}{4\pi
c}\Big|\iiint_{\R^9}\,dvdpdz\int_{|x-y|<ct}\,dy\\
&&\hspace{2cm}{\chi}^{N}(p){\chi}^{N}(z)\frac{f^{N}_{m}(t-c^{-1}|x-y|,y-p-z,v)}{|x-y|^{2}} \frac{(1-c^{-2}|\hat{v}|^2)(\nu\times \hat{v})}{(1+c^{-1}\hat{v}\cdot{\nu})^{2}}\Big|\\
&\leq& \frac{1}{4\pi c}\iiint_{\R^9}\,dvdpdz \int_{|x-y|<ct}\,dy{\chi}^{N}(p){\chi}^{N}(z)\frac{\|f_{0}\|_{L^{{\infty}}(\R^3 \times \R^3)}(4q^{2}(t))^{2}2q(t)}{|x-y|^{2}}\\
&\leq& \frac{M}{c}\int_{|y|<P_0+TQ}\int_{|v|<Q}\frac{1}{|x-y|^{2}}\,dvdy\leq
\frac{M}{c},
\end{eqnarray*}
where \eqref{bdd1} has been used again.
The last term $\mathbb{B}_{3}$ can be shown to be $O(c^{-2})$ in
the same way as $\mathbb{E}_{3}$. 	Now, what is left is $\mathbb{B}_{0}$.
It is easy to calculate that $\partial_{t} B_{0}=-c\nabla\times
E_{0}=0$.  Using Lemma 3.2 and Theorem 2.1, we get

\begin{eqnarray*}
&&\left|\partial_{t}\int_{|\omega|=1}{\frac{t}{4\pi}B^{N}_{m}(0,x+ct\omega)\,d\omega}\right|\\
&=&\left|\partial_{t}\iint_{\R^6}\,dpdz\int_{|\omega|=1}\,d\omega
{\frac{t}{4\pi}{\chi}^{N}(p){\chi}^{N}(z)B_{0}(x-p-z+ct\omega) }\right|\\
&\leq&\iint_{\R^6}\,dpdz \frac{1}{4\pi ct}{\chi}^{N}(p){\chi}^{N}(z)\int_{|\omega|=1}ct|B_{0}(x-p-z+ct\omega)|\,d\omega \\
&&+\iint_{\R^6}\,dpdz {\frac{1}{4\pi}{\chi}^{N}(p){\chi}^{N}(z)\int_{|\omega|=1}ct|\nabla B_{0}(x-p-z+ct\omega)|\,d\omega }\\
&\leq& M\iint_{\R^6} {\frac{1}{4\pi
ct}{\chi}^{N}(p){\chi}^{N}(z)\,dpdz}+\iint_{\R^6}\,dpdz
\int_{|\omega|=1}{\frac{1}{4\pi
}{\chi}^{N}(p){\chi}^{N}(z)ct\frac{1}{c^{2}}\,d\omega } \\
&\leq& \frac{M}{c}.
\end{eqnarray*}
Hence
\begin{equation}
\label{BMP}
B^{N}_{m}=\mathbb{B}_{0}+ \mathbb{B}_{1}+ \mathbb{B}_{2}+
\mathbb{B}_{3}=O(c^{-1}).
\end{equation}
Combing \eqref{EMP} and \eqref{BMP}, we know  that 
\begin{equation}
|E^{N}_{p}-E^{N}_{m}-c^{-1}\hat{v}\times B^{N}_{m}|\leq MG(t)+\frac{M}{c},~~~t<T,
\end{equation}
for $|\hat{v}|<c.$ \\
It remains to estimate $f^{N}_{m}-f^{N}_{p}.$ For ease of notation, we define
$g=f^{N}_{m}-f^{N}_{p}$. It is not difficult to calculate that
\begin{equation}
\begin{split}
&\partial_{t}g+\hat{v}\cdot\nabla_{x}g+(E^{N}_{m}+c^{-1}\hat{v}\times
B^{N}_{m})\cdot\nabla_{v}g\\
 &=(v-\hat{v})\cdot\nabla_{x}f^{N}_{p}+(E^{N}_{p}-E^{N}_{m}-c^{-1}\hat{v}\times
B^{N}_{m})\cdot\nabla_{v}f^{N}_{p}\\
&=\frac{|v|^2\hat{v}}{c^{2}(1+\sqrt{1+c^{-2}|v|^2})}\cdot\nabla_{x}f_p^{N}+(E^{N}_{p}-E^{N}_{m}-c^{-1}\hat{v}\times
B^{N}_{m})\cdot\nabla_{v}f^{N}_{p}
\end{split}
\end{equation}
Note that both $|\nabla_{x}f^{N}_{p}|$ and $|\nabla_{v}f^{N}_{p}|$ are
bounded on $[0,T]\times \mathbb{R}^6$ and $\nabla_{x}f^{N}_{p}(t,x,v)=0$ if
$|v|>q^{N}_{p}(t).$ Hence
\begin{equation}
\begin{split}
&|\partial_{t}g+\hat{v}\cdot\nabla_{x}g+(E^{N}_{m}+c^{-1}\hat{v}\times
B^{N}_{m})\cdot\nabla_{v}g|\\
&\leq \frac{M}{c^2}+M|E^{N}_{p}-E^{N}_{m}-c^{-1}\hat{v}\times B^{N}_{m}|\\
 &\leq \frac{M}{c}+MG(t),~~~0\leq t\leq T.
\end{split}
\end{equation}
For any $x\in \mathbb{R}^3,~v\in \mathbb{R}^3, t\in[0,T],$ we define
$(x(t),v(t))$ as in \eqref{CHC} and calculate
\begin{eqnarray}
\left|\frac{d}{dt}g(t,x(t),v(t))\right|&=&
|\partial_{t}g+\hat{v}\cdot\nabla_{x}g+(E^{N}_{m}+c^{-1}\hat{v}\times
B^{N}_{m})\cdot\nabla_{v}g|\\
&\leq & \frac{M}{c}+MG(t),~~~0\leq t\leq T.
\end{eqnarray}
Note that $g(t,x(t),v(t))|_{t=0}=0$, so $\forall x,v,t$, let $(x(0),v(0))$ be the corresponding initial data of \eqref{CHC}. Then
\begin{eqnarray*}
\left|g(t,x,v)\right|&=&\left|g(t,x(t),v(t))-g(0,x(0),v(0))\right|\\
&&=\left|\int^{t}_{0}\frac{d}{ds}g(s, x(s),v(s))\,ds\right|\\
&&\leq \int^{t}_{0}\left(\frac{M}{c}+MG(s)\right)\,ds\\
&&\leq
\frac{Mt}{c}+\int^{t}_{0}MG(s)\,ds,\quad 0\leq t\leq T.
\end{eqnarray*}
 By the definition of  $g$ and $G(t)$ we get 
$$
G(t)\leq \frac{M}{c}+M\int^{t}_{0}G(s)\,ds,\quad 0\leq t\leq T.
$$
Using the Gronwall's inequality, we get
$$
G(t)\leq \frac{M}{c}\exp(Mt)\leq \frac{M}{c},\quad 0\leq t\leq T.
$$
Therefore
$$\|f^N_{m}-f^N_{p}\|_{L^{\infty}([0,T)\times \R^3 \times \R^3)}+\|E^N_{m}-E^N_{p}\|_{L^{\infty}([0,T)\times \R^3 )}+\|B^N_{m}\|_{L^{\infty}([0,T)\times \R^3 )} \leq \frac{M}{c}.$$
for all  $ c\geq1$. This completes the proof of Theorem.
	\end{enumerate}
\end{proof}

\section*{Acknowledgments}
This work was financially supported by the Effektive Einteilchengleichungen f\"ur korrelierte Vielteilchen-(Coulomb-)Systeme, DFG CH 955/4-1 and  PI 1114/3-1 (2016, 36 months).

%

%

\end{document}